\documentclass{ifacconf}

\usepackage{graphicx}      
\usepackage{natbib}        

\usepackage{amsmath} 
\usepackage{amssymb}  
\usepackage{algorithm}
\usepackage{algorithmicx}
\usepackage{algpseudocode}
\usepackage{mathrsfs}
\usepackage{mathtools}
\let\theoremstyle\relax     
\usepackage{amsthm}         

\usepackage{enumitem}
\setlist[enumerate]{left=0pt, label=(\arabic*)}
\setlist[itemize]{left=0pt}
\usepackage{xcolor}

\newtheorem{theorem}{Theorem}

\newtheorem{lemma}{Lemma}
\newtheorem{remark}{Remark}
\newtheorem{definition}{Definition}
\newtheorem{assumption}{Assumption}
\newtheorem{proposition}{Proposition}

\begin{document}
\begin{frontmatter}

\title{Learning-based Homothetic Tube MPC with Non-Asymptotic Guarantees \thanksref{footnoteinfo}} 

\thanks[footnoteinfo]{The work of Changyi Lei was supported by the Australian Government Research Training Program (RTP) Scholarship and the Australian Research Council under Project DE220101527.  
The work of Dragan Nešić was supported by the Australian Research Council through the Discovery Project under Grant DP250100300. 
The work of Ye Pu was supported by the Australian Research Council under Project DE220101527.}

\author{Changyi Lei, Seth Siriya, Dragan Ne\v{s}i\'{c} and Ye Pu}
\address{Department of Electrical and Electronic Engineering, University of Melbourne, Parkville, VIC 3010, Australia \\
(e-mail: chalei1@student.unimelb.edu.au; siriya@irt.uni-hannover.de; \{dnesic, ye.pu\}@unimelb.edu.au).
}


\begin{abstract}                
This paper studies learning-based MPC for constrained stabilization of discrete-time linear systems with unknown system parameters and additive bounded disturbances. We develop a tractable homothetic-tube MPC scheme in which a high-probability parameter confidence set is generated from non-asymptotic regularized least-squares estimation, rather than assumed a priori. The resulting uncertainty set is embedded into robust tube propagation and constraint tightening, yielding a convex formulation with linear and second-order-cone constraints. We prove high-probability recursive feasibility, robust constraint satisfaction, and input-to-state stability, together with explicit non-asymptotic state bounds. A numerical example illustrates the effectiveness and theoretical guarantees.
\end{abstract}


\end{frontmatter}

\section{Introduction}

Model predictive control (MPC) is a widely used framework for multivariable constrained control, where a finite-horizon optimal control problem is solved repeatedly and implemented in a receding-horizon manner \citep{borrelli2017predictive}. Its performance and safety, however, depend critically on the accuracy of the prediction model \citep{rawlings2017model}. In the presence of disturbances and modeling errors, predicted trajectories may deviate from the actual system behavior, leading to performance degradation and possible constraint violations. This has motivated robust MPC methods that explicitly account for uncertainty \citep{mayne2011robust}, as well as learning-based MPC methods that refine the prediction model from data \citep{hewing2020learning}. This paper focuses on the latter perspective.


Learning-based MPC addresses control problems in which some components of the predictive control formulation are unknown or only partially specified, such as the system dynamics, constraints, cost function, or disturbance model, and seeks to learn or refine these components from data while retaining constraint satisfaction and stability \citep{hewing2020learning}. In this paper, we focus on the case where the system dynamics are unknown. A standard strategy in this setting is to learn or refine a model of the system and then embed this model in a robust MPC formulation, so that residual model mismatch and disturbances are handled through worst-case uncertainty descriptions, tube propagation, and constraint tightening \citep{mayne2005robust,aswani2013provably,di2016indirect,rakovic2012homothetic,rakovic2013homothetic,saccani2023homothetic,rakovic2022homothetic}. Existing approaches differ mainly in the model class used for learning and the associated error quantification. Nonparametric approaches, such as kinky-inference, kernels and Gaussian-process provide flexible modeling methods with associated uncertainty quantification to support robust MPC design \citep{koller2018learning,hewing2019cautious,maddalena2021kpc,limon2017learning}, making them applicable to a broad class of nonlinear systems. However, nonparametric approaches are usually afflicted with scalability, data-coverage, and uncertainty-propagation challenges over the prediction horizon \citep{scampicchio2025gaussian}. On the other hand, Data-Enabled Predictive Control (DeePC) uses measured input-output trajectories of linear systems directly for prediction, without explicitly identifying a model \citep{coulson2019data,berberich2020data}. Extensions include adding terminal ingredients, regularization, and constraint tightening to obtain recursive feasibility, constraint satisfaction, and practical stability under bounded noise \citep{berberich2021terminal,bongard2022robust,kloppelt2025novel}. However, DeePC performance is sensitive to the noise level in the data \citep{sassella2022noise}.

The present work is most closely related to parametric-model-based robust adaptive MPC, where the system is assumed to belong to a parametric model class and the uncertainty set is refined online using data. Among these methods, set-membership estimation (SME) provides a particularly natural route to rigorous constraint satisfaction, since it maintains a feasible parameter set that can be embedded into robust MPC formulations. Early work in this direction includes \cite{tanaskovic2014adaptive}, where set-membership model updates were combined with robust MPC for constrained systems under uniformly bounded model mismatch. Subsequent developments incorporated terminal ingredients and tube-based formulations to obtain recursive feasibility, robust constraint satisfaction, and stability guarantees. In particular, \cite{kohler2019linear,lorenzen2019robust} used homothetic tubes together with set-membership updates to obtain recursive feasibility, robust constraint satisfaction, and stability guarantees. Variants based on recursive least-squares estimation and homothetic tubes were investigated in \cite{zhang2020adaptive}, and related ideas were later extended to nonlinear discrete-time and continuous-time systems in \cite{kohler2021robust,kohler2026certainty,sasfi2023robust}. Beyond passive uncertainty reduction, \cite{parsi2023dual} incorporated predicted future set-membership updates directly into the MPC formulation, enabling exploration to be planned according to its anticipated closed-loop benefit.


One critical limitation remains in existing parametric-model-based robust adaptive MPC schemes. Rigorous recursive feasibility and constraint satisfaction typically rely on the availability of an initial compact parameter set that is assumed to contain the true system before the MPC controller is deployed. Under this assumption, existing methods can provide strong deterministic closed-loop guarantees, including robust constraint satisfaction and stability. However, in a learning-based setting, the validity of such an initial uncertainty set must itself be justified from prior knowledge or from finite data. This certification is nontrivial: classical system-identification results mainly characterize asymptotic consistency \citep{ljung1995system,bai1998convergence}, and often rely, explicitly or implicitly, on stability, stationarity, mixing, or boundedness-type conditions on the data-generating process \citep{yu1994rates}. They therefore do not directly provide a finite-sample uncertainty set suitable for robust MPC initialization, particularly when the data may be generated by an unstable or poorly controlled system.


Recent system identification results provide non-asymptotic estimation error bounds for the estimators, enabling us to bypass the limitations discussed above. Unlike classical consistency results, non-asymptotic theory quantifies the estimation accuracy after a prescribed number of samples and therefore provides an explicit uncertainty radius with high probability guarantees that can be used in learning-based MPC design. For fully observed linear dynamical systems, \cite{simchowitz2018learning} established sharp finite-sample guarantees for ordinary least-squares estimation from a single trajectory, showing that accurate identification is possible without relying on conventional mixing-time arguments \citep{yu1994rates}. Under appropriate excitation and conditioning assumptions, such results yield high-probability bounds of the form
\(
|\hat{\theta}_T-\theta^\star| \leq \epsilon_T,
\)
where the radius $\epsilon_T$ typically decreases at the rate $\tilde O(T^{-1/2})$ after a problem-dependent burn-in period. This line of work has subsequently been extended to unstable and explosive systems \citep{faradonbeh2018finite,sarkar2019near,siriya2024non}, partially observed LTI systems and Markov-parameter estimation \citep{oymak2019non,zheng2020non}, active input design \citep{wagenmaker2020active}, and control-oriented identification for learning-based LQR and robust/adaptive control \citep{dean2020sample}. More recently, non-asymptotic analyses have also been developed for SME \citep{brandle2026beyond,li2024learning}.


Inspired by the above discussion, this paper studies learning-based MPC for constrained linear systems with unknown parameters and additive bounded disturbances. The goal is to design a tractable predictive controller that learns a model from data while maintaining recursive feasibility, robust constraint satisfaction, and closed-loop stability. Unlike robust adaptive MPC schemes that assume a known compact parameter set containing the true system before deployment, we construct the model-uncertainty description from non-asymptotic, high-probability least-squares identification bounds and embed it into a homothetic-tube MPC formulation. Our contributions are as follows:

\begin{itemize}
\item We formulate a learning-based MPC scheme for constrained linear systems with unknown parameters and additive bounded disturbances. The proposed method uses non-asymptotic least-squares identification to learn the uncertain dynamics and employs homothetic-tube MPC to enforce constraints in the presence of residual model uncertainty and disturbances.

\item We construct the parametric uncertainty set from non-asymptotic high-probability identification bounds, rather than assuming a known compact parameter set before MPC deployment. The resulting confidence set contains the true system parameters with a prescribed probability and shrinks as more data are collected, allowing the level of robustness to be calibrated from the available data.

\item We embed the data-dependent confidence set into the homothetic-tube prediction model and derive a tractable MPC reformulation. In particular, the resulting optimization problem can be written as a convex program with a quadratic objective and linear plus second-order-cone constraints.

\item We establish stochastic closed-loop guarantees for the proposed learning-based MPC scheme. On a uniform-in-time high-probability event, the controller is recursively feasible, robustly satisfies the constraints, and stabilizes the closed-loop system. Different from aggregate stability notions commonly established in robust adaptive MPC \citep{lorenzen2019robust}, we derive explicit non-asymptotic pointwise bounds on the closed-loop state and show convergence to a disturbance-dependent neighborhood.

\end{itemize}






The remainder of the paper is organized as below. Sec.\ref{sec:problemformulation} introduces the problem setting and necessary definitions. Sec.\ref{sec:RLS} describes the RLS algorithm for system identification with an estimation error bound and nested properties. Sec. \ref{sec:robustMPC} articulates the proposed MPC algorithm, including the constraint reformulation and terminal set computation. A summary of the algorithm is also provided. Theoretical properties of the closed-loop system are proved in Sec.\ref{sec:theoreticalresults} Sec.\ref{sec:simulation} provides a simulation example to illustrate the efficacy of our method and Sec.\ref{sec:conclusion} concludes the paper. The proofs of all theoretical results are put in the Appendices.

\textit{Notation}: Let $\mathbb{N}$ denote the set of natural numbers, $\mathbb{N}_0 := \mathbb{N}\cup\{0\}$, and $\mathbb{N}_a^b := \{i\in\mathbb{N}\mid a\leq i\leq b\}$. Let $\mathbb{R}^+$ and $\mathbb{R}_{\geq 0}$ denote the sets of positive and non-negative real numbers, respectively. For $A\in\mathbb{R}^{n\times m}$, $\|A\|$, $\|A\|_F$, $A^\intercal$, $A^\dagger$, and $\operatorname{vec}(A)$ denote its spectral norm, Frobenius norm, transpose, Moore--Penrose inverse, and column-wise vectorization, respectively. For $x\in\mathbb{R}^n$, $|x|$, $|x|_\infty$, and $|x|_A:=\sqrt{x^\intercal A x}$ denote the Euclidean, infinity, and weighted norms, respectively. The notation $\operatorname{Unif}(-c,c)$ denotes the uniform distribution on $(-c,c)$ for $c>0$. For points $x^1,\ldots,x^v\in\mathbb{R}^n$, $\mathrm{co}\{x^1,\ldots,x^v\}$ denotes their convex hull. Finally, $\mathbf{1}$ denotes an all-ones vector of compatible dimension.

\section{Problem Formulation} \label{sec:problemformulation}
This paper considers the simultaneous identification and stabilization of the following discrete-time linear system
\begin{equation}
    x_{k+1} = A x_k + B u_k + w_k \label{eq:system}
\end{equation}
with states $x \in \mathcal{X} \subseteq \mathbb{R}^n$, inputs $u \in \mathcal{U} \subseteq \mathbb{R}^m$, additive stochastic disturbance $w \in \mathcal{W} \subset \mathbb{R}^n$, and unknown system matrices $A \in \mathbb{R}^{n \times n}, B \in \mathbb{R}^{n \times m}$. The following assumptions are needed.

\begin{assumption} [Bounded i.i.d Stochastic Disturbance] The disturbance set $\mathcal{W}$ is defined as
\begin{equation}
\label{eq:Wset}
    \mathcal{W} = \{w \in \mathbb{R}^n \ | \ |w|_\infty \leq w_{\max} \}, \ 
\end{equation}
with known $w_{\max}>0$. The disturbance $w_k$ is independent and identically distributed (i.i.d.) following some distribution, and there exists a known constant $\underline{\sigma}_W >0$ such that for all $k\in\mathbb{N}$, we have $ \mathbb{E}[w_k] = 0, \  \mathbb{E}[w_k w_k^\intercal] = \Sigma_w, \ (\Sigma_w)_{ii} \geq \underline{\sigma}_W >0, \ \forall i = 1, \dots, n$.
\label{assump:disturbance}
\end{assumption}


The state–input pair $(x,u)$ is to be confined to the bounded polytope
\begin{equation}
\begin{split}
    \mathcal{Z} &= \bigl\{(x,u)\in \mathcal{X} \times \mathcal{U} \bigm| F\,x + G\,u \le \mathbf{1}\bigr\} \\
    &\subseteq
\bigl\{(x,u) \bigm| |x|_\infty \le x_{\max},|u|_\infty\le u_{\max}\bigr\},  \label{eq:constraints}
\end{split}
\end{equation}
where $F\in\mathbb{R}^{c\times n}$, $G\in\mathbb{R}^{c\times m}$, $0<x_{\max}<\infty, 0<u_{\max}<\infty$. 

The control objectives are:
\begin{enumerate}[label=(\roman*)]
    \item to identify the system parameters online with explicit finite-sample guarantees and; \label{controlobj1}
    \item to find a stabilizing control law for the system such that the constraints \((x,u)\in\mathcal{Z}\) are satisfied robustly, taking into account the parametric uncertainty and the estimated parameters from (i). \label{controlobj2}
\end{enumerate}

\begin{remark}
    In Assumption \ref{assump:disturbance}, the boundedness of disturbance is a necessary assumption for robust control with hard constraints \citep{lorenzen2019robust}. The i.i.d property and variance lower bound will be useful for establishing a persistent excitation (PE) condition and high probability estimation error bound \citep{li2023non}. 
\end{remark}

\section{Regularized Least Squares Parameter Estimation} \label{sec:RLS}
This section explains the regularized least squares (RLS) method for online estimation of the system matrices, which corresponds to the achievement of control objective~\ref{controlobj1}. We first introduce the RLS formulation and derive a high-probability, non-asymptotic bound on the estimation error in closed loop, under an additive probing signal that ensures some persistent excitation (PE) conditions. This bound is then used to construct time-varying parameter sets with a nested property, which will later be exploited in the tube-based MPC design to address control objective~\ref{controlobj2}.

\subsection{Non-asymptotic Identification Error Bound}

Let $\mathcal{D}_T := \{(\bar x_t,\bar u_t,\bar x^+_t)\}_{t=0}^{T-1}$ denote a dataset of $T$ state--input--next-state triples,
where $\bar x_t\in\mathbb{R}^n$, $\bar u_t\in\mathbb{R}^m$, and $\bar x^+_t\in\mathbb{R}^n$. Define the regressor $\bar Z_t := [\bar x_t^\top\ \bar u_t^\top]^\top$ and the linear model
\[
\bar x_t^+ = \theta^{\star\top}\bar Z_t + \bar w_t
\]
where $\theta^*=[A \ B]$ contains the true parameters. Given $\mathcal{D}_T$, the RLS algorithm is used to estimate the parameters $A,B$:
\begin{equation}
    \bar\theta_T
:= \arg\min_{\theta\in\mathbb{R}^{n\times(n+m)}} 
\sum_{t=0}^{T-1}\|\bar x^+_t-\theta^\top \bar Z_t\|_2^2 + \gamma\|\theta\|_F^2,
\label{RLS}
\end{equation}
where $\bar{\theta}_k=[\bar{A}_k \ \bar{B}_k]$ is the estimated parameters at time $k$, $Z_k=[x_k^\intercal \ u_k^\intercal]^\intercal$ is the concatenation of the state and input vector and $\gamma>0$ is the coefficient for the regularization term. 

We show the following non-asymptotic estimation error bound of the RLS method, which will be used in Section \ref{sec:homothetic_tube_constraints} to construct state tubes of the MPC algorithm. The state norm bound $\bar{x}:\mathbb{N}\times\mathbb{R}^n \rightarrow \mathbb{R}_{\geq0}$ is defined as
\begin{equation}
    \bar x(T) := \max_{0\le t\le T} |\bar x_t|,
\label{statebound}
\end{equation}
and the Gramian upper bound is defined as 
\begin{equation}
    \beta_{\max}(T,\gamma) := \sum_{t=0}^{T-1}\big(\bar x(t)^2 + u_{\max}^2\big) + \gamma.
    \label{eq:betamax}
\end{equation}


\begin{lemma}
    For system (\ref{eq:system}) satisfying Assumption \ref{assump:disturbance}, under the RLS algorithm (\ref{RLS}) and control law $u_k=\pi(x_k,\theta)+\eta_k,$ where $\pi(x_k,\theta)$ is any control policy satisfying $|\pi(x,\theta)|_\infty \leq u_{\mathrm{max}}-C$ and $\eta_k \overset{\text{i.i.d.}}{\sim} \operatorname{Unif}([-C,C]^m)$ with $0<C<u_{\max}$. Then for any $\delta \in (0,1)$, there exist computable constants $c_{\mathrm{PE}}>0, p_{\mathrm{PE}}>0$ and $T_{\text{burn-in}}(\delta, c_{\mathrm{PE}}, p_{\mathrm{PE}})$ such that 
    
    (i) for all $x \in \mathbb{R}^n$, $\zeta \in \mathcal{S}^{n+ m-1}$, and $\theta \in \mathbb{R}^{n\times (n+m)}$,
\[
\mathbb{P}\left(
  \Big| \zeta^\top[(x+w);(\pi(x+w,\theta)+v)] \Big|^2
  \geq c_{\mathrm{PE}}
\right) \geq p_{\mathrm{PE}},
\]
where $w \overset{d}{\sim} w_k$ and $v \overset{d}{\sim} \eta_k$, and 

(ii) we have
    \begin{equation}
        \mathbb{P}\left( \lVert \bar{\theta}_T - \theta^*  \rVert \leq \bar\epsilon(T,\delta) \right) \geq 1-\delta, \forall T \geq T_{\text{burn-in}}
    \end{equation}
    where the error bound is defined as 
    \begin{equation}
    \begin{split}
        \bar\epsilon(T,\delta)&:= \tfrac{\gamma^{1/2} \lVert \theta^* \rVert_F}{\sqrt{\tfrac{c_{\mathrm{PE}}p_{\mathrm{PE}}}{4}( T -1)}+\gamma} \\
        & + \tfrac{\sigma_W \sqrt{2n\ln(\tfrac{3n}{\delta})+\tfrac{n+m}{2}\ln(\tfrac{\beta_{\text{max}}(T,\gamma)}{\gamma})} }{\sqrt{\tfrac{c_{\mathrm{PE}}p_{\mathrm{PE}}}{4}(T-1)}+\gamma}
    \end{split} \label{eq:estimationerror}
    \end{equation}
    and $\gamma$ is from \eqref{RLS}.
    \label{lem:RLSbound}
\end{lemma}

\begin{remark}
    Lemma~\ref{lem:RLSbound} serves two purposes. First, it verifies a PE condition induced by the additive probing signal. Second, after a computable burn-in time, it gives a finite-sample high-probability error bound for the RLS estimator. This differs from classical asymptotic consistency results for ordinary least squares, which do not directly provide an explicit uncertainty radius at a finite sample size. The quantity \(\beta_{\max}(T,\gamma)\) upper bounds the size of the regularized data Gramian and captures the effect of the observed state and input magnitudes on the estimation error. In particular, when the trajectory is bounded, \(\beta_{\max}(T,\gamma)\) grows at most linearly in \(T\), so the logarithmic term in \eqref{eq:estimationerror} grows slowly and the dominant rate of \(\bar\epsilon(T,\delta)\) is of order \(\tilde O(T^{-1/2})\).
\end{remark}

We divide the collected data into an offline phase and an online phase. During the offline phase, \(k_0\) data pairs are collected under a bounded exploratory policy
\(
u_k=\pi_0(x_k)+\eta_k, \ 
\eta_k \overset{\mathrm{i.i.d.}}{\sim}\operatorname{Unif}([-C,C]^m),\ k<k_0 .\)
The index \(k\in\mathbb{N}\) denotes the total number of collected data, and the online MPC controller is deployed from time \(k=k_0\). The RLS estimator is initialized and updated using the data collected across these phases. To apply the non-asymptotic identification result uniformly after deployment, we impose the following boundedness condition.

\begin{assumption}
\label{assump:boundedoffline}
The offline state trajectory is bounded, i.e.,
\(
\bar x(k_0):=\max_{0\leq k<k_0}|\bar x_k|
\leq x^{\mathrm{off}}_{\max},
\)
for some \(x^{\mathrm{off}}_{\max}\geq 0\). Moreover, \(k_0\geq T_{\mathrm{burn\text{-}in}}\).
\end{assumption}
\begin{remark}
Assumption~\ref{assump:boundedoffline} ensures that the finite-sample estimation radius is well defined at the beginning of the online phase. The constant \(x^{\mathrm{off}}_{\max}\) may be chosen after the offline data are collected and therefore need not be specified a priori. For stable linear systems with bounded disturbances, this boundedness requirement is automatically satisfied over any finite data horizon. For unstable systems, it can be enforced by collecting multiple finite trajectories, for instance by repeatedly resetting the system inside a compact set \citep{tu2024learning}.
\end{remark}

\begin{proposition}[Online non-asymptotic identification radius]
\label{prop:bareps_k}
Suppose Assumptions~\ref{assump:disturbance} and~\ref{assump:boundedoffline} hold, and let the RLS estimate \(\bar\theta_k=[\bar A_k\ \bar B_k]\) be generated by \eqref{RLS}. Define $x_{\max}^{\text{id}}=\max\{ x_{\max}, x^{\mathrm{off}}_{\max}\}$. Assume that \((x_k,u_k)\in\mathcal Z\) for all \(k\ge k_0\). Then, for any \(\delta\in(0,1)\), we have
\begin{equation}
\mathbb{P}(|\bar\theta_k-\theta^\star|
\le \bar\epsilon_k,
\ \forall k\ge k_0) \geq 1-\delta,
\end{equation}
where
\begin{equation}
    \begin{split}
        \bar\epsilon_k &:= \bar\epsilon(k,\delta) \\
        &= \tfrac{\gamma^{1/2} \lVert \theta^* \rVert_F}{\sqrt{\tfrac{c_{\mathrm{PE}}p_{\mathrm{PE}}}{4}( k -1)}+\gamma} \\
        & + \tfrac{\sigma_W \sqrt{2n\ln(\tfrac{3n}{\delta})+\tfrac{n+m}{2}\ln(1+\tfrac{k[(x_{\max}^{\text{id}})^2+u_{\max}^2]}{\gamma})} }{\sqrt{\tfrac{c_{\mathrm{PE}}p_{\mathrm{PE}}}{4}(k-1)}+\gamma}.
    \end{split} \label{eq:bareps_k}
    \end{equation}
\end{proposition}


\subsection{Uncertainty-set Construction and Update}

In this subsection, we convert the non-asymptotic identification radius into a nested uncertainty set that contains the true system with high probability. For the subsequent closed-loop analysis, this uncertainty description must satisfy two key properties: high-probability containment of the true system parameters and monotonic shrinkage over time. The containment property enables robust constraint satisfaction with respect to the true system, while the nesting property is essential for recursive feasibility and stability proof. To ensure these properties, we construct the uncertainty set used by MPC through an outer approximation followed by an intersection update.

    The non-asymptotic bound \eqref{eq:bareps_k} enables the construction of vanilla uncertain parameter sets as 
\begin{equation}
    \bar\Theta_k = \{ (\bar{A}, \bar{B}) \mid \|\bar{A}-\bar{A}_k \|_2 \leq \bar\epsilon_k, \|\bar{B}-\bar{B}_k \|_2 \leq \bar\epsilon_k\}.
\end{equation}
We calculate a tight outer approximation of $\Theta_k$ using hyper-rectangles as 
\[
\label{eq:Theta_rect_k}
\begin{aligned}
    \bar{\Theta}_k^{\mathrm{p}}
=\Big\{(\bar A,\bar B) \mid & 
\bar A_k-\bar\epsilon_k \mathbf 1_{n\times n} \le \bar A \le \bar A_k+\bar\epsilon_k \mathbf 1_{n\times n},\\
&\bar B_k-\bar\epsilon_k \mathbf 1_{n\times m} \le \bar B \le \bar B_k+\bar\epsilon_k \mathbf 1_{n\times m}
\Big\},
\end{aligned}
\]
which can be equivalently written as
\[
\bar{\Theta}_k^{\mathrm{p}}
=\Big\{(\bar A,\bar B) \mid 
H_A\,\mathrm{vec}(\bar A)\le h^k_A, 
H_B\,\mathrm{vec}(\bar B)\le h^k_B
\Big\},
\]
where
\[
H_A =
\begin{bmatrix}
I_{n\times n}\\
- I_{n\times n}
\end{bmatrix},\ 
h^k_A =
\begin{bmatrix}
\mathrm{vec}(\bar A_k)+\epsilon_k \mathbf 1_{n\times n}\\
-\mathrm{vec}(\bar A_k)+\epsilon_k \mathbf 1_{n\times n}
\end{bmatrix},
\]
\[
H_B =
\begin{bmatrix}
I_{n\times m}\\
- I_{n\times m}
\end{bmatrix},
\ 
h^k_B =
\begin{bmatrix}
\mathrm{vec}(\bar B_k)+\epsilon_k \mathbf 1_{n\times m}\\[2pt]
-\mathrm{vec}(\bar B_k)+\epsilon_k \mathbf 1_{n\times m}
\end{bmatrix}.
\]
For each time step $k$, we construct the following parameter set
\begin{equation}
\label{eq:Theta_k}
\Theta_k = \begin{cases}
    \bar{\Theta}_k^{\mathrm{p}} \cap \Theta_{k-1}, & k > k_0 \\
    \bar{\Theta}_k^{\mathrm{p}}, & k = k_0
\end{cases}.
\end{equation}

Representing $\Theta_k$ in the compact form
\[
\Theta_k=\{(\bar A, \bar B)\mid H_k \; [\mathrm{vec}(\bar A)^\top, \mathrm{vec}(\bar B)^\top]^\top\le h_k\}.
\]
Denote the vertex set of $\Theta_k$ as
$\mathcal{V}(\Theta_k)=\{\theta_k^{(1)},\dots,\theta_k^{(N_v)}\}$.
We define the point estimate as the arithmetic mean of all vertices,
\begin{equation}
\label{eq:hattheta_k}
\hat{\theta}_k
=[\hat{A}_k, \hat{B}_k]= \tfrac{1}{N_v}\sum_{j=1}^{N_v}\theta_k^{(j)}.
\end{equation}
In this hyper-rectangular case $\hat{\theta}_k$ is equal to the Chebyshev center of $\Theta_k$.

\begin{remark}
    We choose to construct $\Theta_k$ based on intersection, which will be useful for the proof of nested properties of the parameter sets, and therefore contribute to the recursive feasibility of MPC. The outer approximation using hyper-rectangles are for easy computation of those intersection operations. The parameter set $\Theta_k$ in \eqref{eq:Theta_k} will be used to construct iterative tubes for the MPC algorithm to enforce robust constraints satisfaction, while the point estimate $\hat{\theta}_k$ will be used to implement CE MPC to ensure stability.
\end{remark}

We have the following properties for $\Theta_k$ and $\hat{\theta}_k$.
\begin{lemma}
    Suppose Assumptions \ref{assump:disturbance} -- \ref{assump:boundedoffline} are satisfied. Consider $\Theta_k$ in \eqref{eq:Theta_k} and $\hat{\theta}_k$ in \eqref{eq:hattheta_k}, we have
    
    \text{(i)} \( \theta^* \in \Theta_k, \Theta_k \subseteq \Theta_{k-1}, \hat{\theta}_k \in \Theta_k \);
    
    \text{(ii)} \( \epsilon_k \leq \sqrt{n \times (n+m)}\; \bar{\epsilon}_k \);
    
    for all $k \geq T_{\mathrm{burn-in}}$ with probability at least $1-\delta$.
    \label{lem:nestedparamsets}
\end{lemma}
Lemma \ref{lem:nestedparamsets} has several important implications for the MPC control design. Firstly, by showing that the set is shrinking, the uncertainty of the system keeps decreasing and so is the conservatism of the tube construction. Secondly, the constructed uncertainty set is monotonically shrinking, and thus enables recursive feasibility with uniform-in-time high probability.

\section{Learning-based Homothetic Tube MPC} \label{sec:robustMPC}
To achieve the control objective \ref{controlobj2}, this section proposes a tube-based MPC algorithm utilizing online parameter estimation and uncertainty sets constructed in Section \ref{sec:RLS}. The nominal prediction of MPC is performed employing $\hat{A}_k, \hat{B}_k$ in a certainty equivalence manner. For the satisfaction of constraints, we use $\Theta_k$ in \eqref{eq:Theta_k} to determine a sequence of admissible state tubes $\{ \mathbb{X}_{l|k} \}_{l \in \mathbb{N}_0^N}$ that provides an outer approximation of the predicted states that satisfies (\ref{eq:constraints}). Specifically, for any $k \geq k_0$, given the parametric set $\Theta_k$ and constraint set $\mathcal{Z}$, let
\begin{subequations}\label{eq:tubeconstraints}
\begin{align}
\mathbb{X}_{0|k} & \ni x_k , \label{eq:tubeconstraintsa} \\[0.3em]
\mathbb{X}_{l+1|k} &\ni \hat{A}x + \hat{B} u_{l|k}(x) + w ,
\label{eq:tubeconstraintsb} \\
&\quad \forall x \in \mathbb{X}_{l|k},\ w \in \mathcal{W},\
\forall [\hat{A} \ \hat{B}] \in \Theta_k \nonumber \\[0.3em]
x \times u_{l|k}(x) &\in \mathcal{Z} ,
\ \forall x \in \mathbb{X}_{l|k} ,
\label{eq:tubeconstraintsc}
\end{align}
\end{subequations}

with some input parameterization $u_{l|k}:\mathbb{R}^n \rightarrow \mathbb{R}^m, \forall l \in \mathbb{N}_0^N$. In principle, a receding horizon feedback law can be determined by solving the optimization problem
\begin{equation}
\begin{split}
    V_N(x_k, \hat{\theta}_k, \epsilon_k) 
&= \min_{\mathbf{u}_{N|k}} J(x_k, \hat{\theta}_k, \mathbf{u}_{N|k}) 
\\
\text{s.t.} \quad & (\ref{eq:tubeconstraints}), \ \mathbb{X}_{N|k} \subseteq \mathbb{X}_f
\end{split} \label{eq:MPCproblem1}
\end{equation}
with a user-defined cost function $J$, prediction horizon $N$, and terminal set $\mathbb{X}_f$. Denote the minimizer by $\mathbf{u}^*_{N|k} = \{ u^*_{l|k} \}_{l \in \mathbb{N}_0^{N-1}}$, and then the final online control law is defined by 
\begin{equation}
u(x_k) = u^*_{0|k} + \eta_k, \eta_k \overset{\text{i.i.d.}}{\sim} \operatorname{Unif}([-\eta_{max},\eta_{max}]^m)
    \label{eq:controllaw}
\end{equation}
where $\eta_k$ is an excitory signal injected to ensure the PE condition and $\eta_{max}>0$ is a user-defined excitation magnitude. Thus, it also lies inside a polytopic set 
\begin{equation}
\label{eq:Hset}
    \mathcal{H} = \{\eta \ | \ H\eta \leq \eta_{max} \mathbf{1} \},
\end{equation}
and the element-wise variance is $\sigma^2_\eta = \tfrac{1}{3}\eta_{max}^2$. 

In the remainder of this section, the design of each component in \eqref{eq:MPCproblem1} will be articulated. Importantly, state-input constraints in the form of \eqref{eq:tubeconstraints} are intractable, resulting in a nonconvex, min-max optimization problem \eqref{eq:MPCproblem1}. A reformulation of \eqref{eq:tubeconstraints} using homothetic tube parameterization will be explained, which converts \eqref{eq:MPCproblem1} into a convex optimization problem.

\subsection{Homothetic tube-based state-input constraints} \label{sec:homothetic_tube_constraints}

This section presents a homothetic tube parameterization which results in a tractable reformulation of \eqref{eq:MPCproblem1}. The main advantage of this parameterization is to convert the prediction constraints (\ref{eq:tubeconstraints}) into polyhedral and second-order cone (SOC) constraints. 

We make the following assumption.
\begin{assumption} \cite{lorenzen2019robust}
    There exists control gain matrix $K \in \mathbb{R}^{m\times n}$ and positive definite matrix \( P \in \mathbb{R}^{n \times n} \) satisfying
\begin{equation}
(\hat{A}+\hat{B}K)^\intercal P (\hat{A}+\hat{B}K) + Q + K^\intercal R K \preceq P
\end{equation}
for all $[\hat{A}, \hat{B}] \in \Theta_{k_0} $ where $k_0$ is the starting time instant of the online phase.
\label{assump:KP}
\end{assumption}

\begin{remark}
Assumption~\ref{assump:KP} can be ensured by a certainty-equivalence LQR when the system identification error is sufficiently small; see the perturbation analysis in \citep{mania2019certainty}. 
\end{remark}
The input is parameterized as 
\begin{equation}
\label{eq:controlparameterization}
    u_{l|k}(x)=Kx+v_{l|k},
\end{equation}
with decision variables $\mathbf{v}_{N|k} = \{ v_{l|k} \}_{l \in \mathbb{N}_0^{N-1}}, \ v_{l|k} \in \mathbb{R}^m$ and pre-stabilizing feedback gain \( K \in \mathbb{R}^{m \times n} \) satisfying Assumption \ref{assump:KP}.

In the following, homothetic tube parameterization is used to recast the original MPC problem \eqref{eq:MPCproblem1} (which is a nonconvex, bilevel optimization problem) into a convex optimization program with second-order-cone and linear constraints inspired by \citep{lorenzen2019robust}. This is achieved by restricting the sets $\mathbb{X}_{l|k}$ to translations and scaling of a convex base polytope $\mathbb{X}_0$. For a given convex set
\begin{equation} \label{eq:basepolytope}
    \mathbb{X}_0 = \{ x \in \mathbb{R}^n \mid H_x x \leq \mathbf{1} \}
\end{equation}
with vertices \( \{ x^1, \ldots, x^v \} \) and matrix $H_x \in \mathbb{R}^{h \times n}$, decision variables \( \mathbf{z}_{N|k} = \{ z_{l|k} \}_{l \in \mathbb{N}_0^N}, z_{l|k} \in \mathbb{R}^n \), and 
\( \boldsymbol{\alpha}_{N|k} = \{ \alpha_{l|k} \}_{l \in \mathbb{N}_0^N}, \alpha_{l|k} \in \mathbb{R}_{\geq 0} \), for $k\in \mathbb{N}$, define
\begin{equation}
\begin{split}
    \mathbb{X}_{l|k} &= \{ z_{l|k} \} \oplus \alpha_{l|k} \mathbb{X}_0 
\\
& = \left\{ x \in \mathbb{R}^n \;\middle|\; H_x (x - z_{l|k}) \leq \alpha_{l|k} \mathbf{1} \right\} \\
&= \{ z_{l|k} \} \oplus \alpha_{l|k} \, \text{co}\{ x^1, x^2, \ldots, x^v \}.
\end{split}
\label{eq:homotube}
\end{equation}

\begin{remark}
    The base polytope \eqref{eq:basepolytope} is user-defined, and should be compact. The scaling variable $\alpha$ in \eqref{eq:homotube} should be non-negative. Intuitively, $z$ is the center of the polytope and $\alpha$ is a scaling coefficient that determines its size.
\end{remark}

\begin{proposition} [constraint reformulation] For system \eqref{eq:system}, let $\{\mathbb{X}_{l|k}\}_{l \in \mathbb{N}_0^{N}}$ be parametrized as in (\ref{eq:homotube}), and consider decision variables $\mathbf{z}_{N|k}$, $\boldsymbol{\alpha}_{N|k}$, $\mathbf{v}_{N|k}$ and estimation error bound $\epsilon_k$ in Lemma \ref{lem:nestedparamsets} (ii). Given $\hat{A}_k,\hat{B}_k$ in \eqref{eq:hattheta_k} and matrix $K$ satisfying Assumption \ref{assump:KP}. Eqs. (\ref{eq:tubeconstraintsa})-(\ref{eq:tubeconstraintsc}) are satisfied if and only if the following holds for all $j \in \mathbb{N}^{v}_1$, $l \in \mathbb{N}_0^{N-1}$:
\begin{subequations} \label{eq:reconstraints}
\begin{align}
&(F + GK)z_{l|k} + G v_{l|k} + \bar{\eta} + \alpha_{l|k} \bar{f} \leq \mathbf{1} \label{eq:reconstraints_a} \\ 
&-H_x z_{0|k} - \alpha_{0|k} \mathbf{1} \leq -H_x x_k \label{eq:reconstraints_b} \\
&H_x \left( (\hat{A}_k + \hat{B}_k K)(z_{l|k} + \alpha_{l|k} x^j) + \hat{B}_k v_{l|k} \right) \nonumber\\
&- H_x z_{l+1|k} - \alpha_{l+1|k} \mathbf{1} + \bar{\eta}_{\hat{B}_k}
\leq -\bar{w} - \epsilon_k \nonumber \\
& \times \Bigl(|z_{l|k}+\alpha_{l|k}x^j|
+|K(z_{l|k}+\alpha_{l|k}x^j)+v_{l|k} | \nonumber\\
& \quad + \sqrt{m}\,\eta_{max} \mathbf{1} \Bigr)
\begin{pmatrix}
|H_x^1|\\
\vdots\\
|H_x^h|
\end{pmatrix} \label{eq:reconstraints_c}
\end{align}
\end{subequations}
where 
\begin{equation}
    \left[\bar{\eta}_{\hat{B}_k} \right]_i=\max_{\eta \in \mathcal{H}} \left( [H_x]_i \hat{B}_k \eta  \right), \forall i \in \mathbb{N}_1^h
    \label{eq:temp3}
\end{equation}
\begin{equation}
    [\bar{w}]_i := \max_{w \in \mathcal{W}} [H_x]_i w, \forall i \in \mathbb{N}_1^{h}
    \label{eq:temp4}
\end{equation}
\begin{equation}
    [\bar{f}]_i := \max_{x \in \mathbb{X}_0} [F + GK]_i x, \forall i \in \mathbb{N}_1^{c}
    \label{eq:temp5}
\end{equation}
\begin{equation}
    [\bar{\eta}]_i := \max_{\eta \in \mathcal{H}} [G]_i \eta, \forall i \in \mathbb{N}_1^{c}
    \label{eq:temp6}
\end{equation}
with \( F, G \) being defined in (\ref{eq:constraints}) and $\eta_{\max}$ in \eqref{eq:controllaw}.

\label{prop:predictiontube}
\end{proposition}

\begin{remark}
    The value of Proposition \ref{prop:predictiontube} is that it transforms the original constraints with respect to the decision variables into polytopic constraints \eqref{eq:reconstraints_a} -- \eqref{eq:reconstraints_b} and SOC constraints \eqref{eq:reconstraints_c}, which can be solved using convex solvers.
\end{remark}

\subsection{Cost Function}
We use the point estimate $\hat{A}_k,\hat{B}_k$ in \eqref{eq:hattheta_k} to define a certainty equivalent cost, where the predicted states and inputs are calculated using $\hat{A}_k,\hat{B}_k$. Let \( Q \in \mathbb{R}^{n \times n} \) and \( R \in \mathbb{R}^{m \times m} \) be positive definite cost matrices. Using the matrices $K,P$ from Assumption \ref{assump:KP}, the finite horizon MPC cost function is then given by
\begin{equation}
\begin{aligned}
    J_N&(x_k, [\hat{A}_k \ \hat{B}_k], \mathbf{v}_{N|k}) \\
&= \sum_{l=0}^{N-1} \left( | \hat{x}_{l|k} |_Q^2 + | \hat{u}_{l|k} |_R^2 \right) 
+ | \hat{x}_{N|k} |_P^2
\end{aligned} \label{eq:costfunction_CEMPC}
\end{equation}
where \( \hat{x}_{l|k}, \hat{u}_{l|k} \) are defined recursively by
\begin{equation}
\begin{aligned}
\hat{x}_{l+1|k} &= \hat{A}_k\hat{x}_{l|k} + \hat{B}_k\hat{u}_{l|k}, \\
x_{0|k} &= x_k, \\
\hat{u}_{l|k} &= K \hat{x}_{l|k} + v_{l|k}.
\end{aligned}
\label{CEtrajectory}
\end{equation}

\subsection{Terminal Constraints}
Similar to the standard practice in tube-based robust MPC design, we employ stabilizing terminal constraints for stability and recursive feasibility proof. The following assumption describes the desired terminal set, and a recursive algorithm is presented to determine a robust invariant terminal set considering the dynamics of the translation and scaling variables $(z,\alpha)$. Such a recursive calculation is performed at initial timestep using the initial parameter estimate and its probabilistic error bound (i.e. $(\hat{A}_0, \hat{B}_0, \epsilon_0, \delta)$) derived from Lemma \ref{lem:RLSbound}.
\begin{assumption} \citep{lorenzen2019robust}
Under control law \eqref{eq:controllaw} and \eqref{eq:controlparameterization}, there exists a non-empty terminal set 
\[
\mathbb{X}_{f} = \left\{ (z, \alpha) \in \mathbb{R}^n \times \mathbb{R}_{\geq 0} 
\;\middle|\; H_T z + h_T \alpha \leq \mathbf{1} \right\}
    \]
such that \( (x, Kx) \in \mathcal{Z} \) for all 
\( x \in \{z\} \oplus \alpha \mathbb{X}_0 \), 
\( (z, \alpha) \in \mathbb{X}_f \), 
and 
\[
(z, \alpha) \in \mathbb{X}_{f} 
\;\Longrightarrow\; 
\exists (z^+, \alpha^+) \in \mathbb{X}_{f} \text{ s.t.}
\]
\[
(\hat{A}+\hat{B} \,K)(\{z\} \oplus \alpha \mathbb{X}_0) 
\oplus \mathcal{H} \oplus \mathcal{W} 
\subseteq \{z^+\} \oplus \alpha^+ \mathbb{X}_0
\]
for all $[\hat{A}, \hat{B}] \in \Theta_{k_0}$.
\label{assump:terminalset}
\end{assumption}

Algorithms similar to \citep{lorenzen2019robust} and \citep{rakovic2013homothetic} can be used to explicitly calculate a terminal set $\mathbb{X}_{f}$ that satisfies Assumption \ref{assump:terminalset}.

\subsection{MPC problem and algorithm}

The final MPC problem that we are solving is 
\begin{equation} \label{eq:MPCproblemfinal}
\begin{split}
    V_N(x_k, \hat{\theta}_k, \epsilon_k) 
&= \min_{\mathbf{z}_{N|k}, \mathbf{\alpha}_{N|k}, \mathbf{v}_{N|k}} J(x_k, \hat{\theta}_k, \mathbf{v}_{N|k}) \\
&\text{s.t.} \ (\ref{eq:reconstraints}), \  \mathbb{X}_{N|k} \subseteq \mathbb{X}_f.
\end{split}
\end{equation}
with the certainty equivalence cost function $J$ defined in \eqref{eq:costfunction_CEMPC}. In contrast to \eqref{eq:MPCproblem1}, problem \eqref{eq:MPCproblemfinal} is computationally tractable and convex, due to the homothetic tube parameterization and second order cone constraint reformulation. Thus the control objective \ref{controlobj2} can be addressed by any conic program solver.

The final control algorithm is summarized in Algorithm \ref{MPCalgorithm}. We keep track of the most recent feasible step $k_{\mathrm{feas}}$. If at some step $k$, the problem is not feasible, then we reuse the feasible solutions at step $k_{\mathrm{feas}}$; otherwise, $k=k_{\mathrm{feas}}$. For convenience of the following derivation, we denote the first element of the optimal sequence $\mathbf{v}^*_{N|k}$ to be $v^*_{0|k}$, and the control policy at time $k$ can be represented by $\alpha(x_k, \eta_k, \hat{\theta}_k)$.

\begin{algorithm}
\caption{Learning-based Homothetic Tube MPC}
\label{MPCalgorithm}
\begin{algorithmic}[1]

\State \textbf{Offline:}
\State Input: $k_0, N_0, \eta_{max},x_{\max},u_{\max},w_{\max},\delta$, prediction horizon $N$, disturbance set $\mathcal{W}$, constraint set $\mathcal{Z}$, cost coefficients $Q,R$
\State Collect $k_0$ data pairs
\State Calculate $(\hat{A}_0, \hat{B}_0)$ from \eqref{RLS}. Calculate $\epsilon_0$ from Lemma \ref{lem:RLSbound}. Calculate $K, P$ satisfying Assumption \ref{assump:KP}.
\State Calculate terminal set $\mathbb{X}_{f}$ satisfying Assumption \ref{assump:terminalset}

\vspace{0.5em}
\State \textbf{Online:}

\State Solving the optimization problem \eqref{eq:MPCproblemfinal}
\State Store optimal sequence $\{v_{i|k}^*\}_{i=0}^{N-1}$
\State Sample $\eta_0 \sim \operatorname{Unif}([-\eta_{max},\eta_{max}]^m)$ and apply $u_0 = Kx_0 + v_{0|0}^* + \eta_0$ 
\State Set feasibility flag $k_{\mathrm{feas}} \gets 0$

\For{$k = k_0, k_0+1, \dots$}
    \State Measure $x_k, u_k$
    \State Update parameter estimates $(\hat{A}_k, \hat{B}_k)$ using (\ref{RLS})
    \State Try solving the optimization problem \eqref{eq:MPCproblemfinal}
    \If{a feasible solution exists}
        \State Store optimal sequence $\{v_{i|k}^*\}_{i=0}^{N-1}$
        \State Update $k_{\mathrm{feas}} \gets k$
    
    \EndIf
    \State Sample $\eta_k \sim \operatorname{Unif}([-\eta_{max},\eta_{max}]^m)$
    \If{$k-k_{\mathrm{feas} }\leq N$}
        \State Apply $u_k = Kx_k + v_{k-k_{\mathrm{feas}}|k_{\mathrm{feas}}}^* + \eta_k$ 
    \Else
        \State Apply $u_k = Kx_k + \eta_k$ 
    \EndIf
\EndFor
\end{algorithmic}
\end{algorithm}

\section{Theoretical Results} \label{sec:theoreticalresults}

In this section, we will show the recursive feasibility and stability properties of the closed-loop system with high probability. Based on Lemma \ref{lem:nestedparamsets}, we can show recursive feasibility and constraints satisfaction as below.
\begin{proposition}[Recursive Feasibility and Constraint Satisfaction] 
Consider system \eqref{eq:system} under Assumptions~\ref{assump:disturbance} -- \ref{assump:terminalset}. If $D(x_{k_0},\Theta_{k_0})\neq\emptyset$, then with probability at least $1-\delta$ the following statements hold for all $k\ge k_0$:
\begin{itemize}[leftmargin=0.7cm,labelindent=0cm]
\item[(i)]
$\theta^\ast\in\Theta_{k+1}$,
$\Theta_{k+1}\subseteq\Theta_k$,
$\hat\theta_{k+1}\in\Theta_{k+1}$;
\item[(ii)]
$\mathbb D(x_{k+1},\Theta_{k+1})\neq\emptyset$;
\item[(iii)]
$(x_{k+1},u_{k+1})\in\mathcal Z$.
\end{itemize} \label{proposition:feasibility}
\end{proposition}

\begin{definition}[Stochastic Practical ISS function]\label{def:SLF} 
Define the system evolution as $g(x,u,w):=Ax+Bu+w$. A Borel measurable function $V : \mathcal{X}_{\mathrm{RPI}}\times 
\Theta \times \mathbb{R}_{\ge 0} \to \mathbb{R}_{\ge 0}$ 
over the domain $\mathcal{X}_{\mathrm{RPI}} \subseteq \mathcal{X}$ is said to be a
stochastic practical ISS function with the estimation error bound 
$\bar{\theta}$, for the system $g$ under the influence of the policy $\pi$ 
and the process noise $w$, if  

1) $g(x, \pi(x,\eta, \tilde{\theta}_1+\theta^*), w) \in \mathcal{X}_{\mathrm{RPI}}$; 

2) there exist $\alpha_1,\alpha_2,\alpha_3 \in \mathcal{K}_\infty$, 
$\sigma_3 \in \mathcal{K}$, and $\tilde d \ge 0$, such that  
\begin{equation}
\label{eq:stoch_lya_ineq}
\begin{split}
    &\alpha_1(|x|) \; \le\; V(x,\tilde{\theta}_1+\theta^*,\|\tilde{\theta}_1\|_2) \;\le\; \alpha_2(|x|), \\
& \mathbb{E}\!\left[\,V\bigl(g(x,\pi(x,\eta,\tilde{\theta}_1+\theta^*),w),\tilde{\theta}_2+\theta^*, \|\tilde{\theta}_2\|_2\bigr)\right]\\
& -V(x,\tilde{\theta}_1+\theta^*,\|\tilde{\theta}_1\|_2) \le -\alpha_3(|x|) + \tilde d + \sigma_3(\|\tilde{\theta}_1\|_2);
\end{split}
\end{equation}

3) $\alpha_3 \circ \alpha_2^{-1}$ is lower bounded 
by some convex function in $\mathcal{K}_\infty$ with $ \eta \overset{d}{\sim} \eta_k,  w \overset{d}{\sim} w_k$ for all $x \in \mathcal{X}_{\mathrm{RPI}}$ and
$\tilde{\theta}_1,\tilde{\theta}_2 \in \mathcal{B}_{\bar{\theta}}(0)$.
\end{definition}

We show that system (\ref{eq:system}) under Algorithm \ref{MPCalgorithm} admits a stochastic ISS function. Denote the feasible solution set to be 
\[
    \mathbb{D}(x_k, \Theta_{k}) =  \left\{(\mathbf{z}_{N|k}, \mathbf{\alpha}_{N|k}, \mathbf{v}_{N|k} )\mid (\ref{eq:reconstraints}), \mathbb{X}_{N|k} \subseteq \mathbb{X}_f\right\}
\]
at timestep $k$ where $\epsilon>0$ is the parametric uncertainty magnitude used in (\ref{eq:reconstraints}). 

\begin{proposition}[Existence of a stochastic practical ISS function]
Consider  the closed-loop system \eqref{eq:system} under Algorithm \ref{MPCalgorithm} with Assumptions \ref{assump:disturbance}-\ref{assump:terminalset} and given online phase starting time instant $k_0$. If $\mathbb D(x_{k_0},\Theta_{k_0})\neq\emptyset$, then the optimal finite-horizon cost
$V_N(x,\hat\theta,\epsilon)$ defined in \eqref{eq:MPCproblemfinal} is a stochastic practical ISS Lyapunov function
for all $k\ge k_0$ with probability $1-\delta$. Specifically, the following hold:
\begin{enumerate} 

\item (RPI domain) The set \[\mathcal{X}_{\mathrm{RPI}}=\left\{ x \in \mathbb{R}^n  \, \Big| \, |x| \leq \max \left\{|x_{k_0}|,  \bar x_{\mathrm{RPI}} \right\} \right\}\] is RPI for the closed-loop system, where \(\bar x_{\mathrm{RPI}}:= \tfrac{4 \epsilon_{k_0} ^2+ c_A (|w_{\max}|^2 + |\eta_{\max}|^2)}{c_3}\) for some constant $c_3>0$;

\item (Sandwich bounds) There exist $\alpha_1,\alpha_2 \in \mathcal{K}_\infty$ such that
\[
\alpha_1(|x|) \leq V_N(x,\hat\theta,\epsilon) \leq \alpha_2(|x|).
\]

\item (Stochastic drift condition) there exist $\alpha_3 \in \mathcal{K}_\infty$, 
$\sigma_3 \in \mathcal{K}$, and $\tilde d \ge 0$ such that
\begin{equation}
\begin{split}
\mathbb E&\!\left[V_N\!\left(g(x_k,\alpha(x_k,\eta,\hat\theta_k),w),
\hat\theta_{k+1},\epsilon_{k+1}\right) \right]
\\
& - V_N(x_k,\hat\theta_k,\epsilon_k)
\le -\alpha_3(|x_k|) + \tilde d + \sigma_3(\epsilon_k)
\end{split}
\end{equation}

\end{enumerate}
with $ \eta \overset{d}{\sim} \eta_k,  w \overset{d}{\sim} w_k$ and for all $\epsilon_k, \epsilon_{k+1} \leq \epsilon_{k_0}$.
\label{prop:SLya}
\end{proposition}

\begin{theorem} (High Probability Stability Bound)
    Consider system \eqref{eq:system} under Algorithm \ref{MPCalgorithm} with Assumptions \ref{assump:disturbance}--\ref{assump:terminalset}. If $\mathbb D(x_{k_0},\Theta_{k_0})\neq\emptyset$, then there exist $c_4>0$ and class-$\mathcal{L}$ function $\beta(k)$ such that 
    \begin{equation}
        \mathbb{P}(|x_k| \leq \beta(k) + c_4) \geq 1-\delta, \forall k  \geq k_0.
    \label{eq:statebound}
    \end{equation}
    
    \label{theorem:statebound}
\end{theorem}
\begin{proof}
    We apply \citep[Corollary 1]{siriya2025frameworkadaptivestabilisationnonlinear} to obtain the final conclusion. Note that although our definition of the stochastic ISS function in Definition \ref{def:SLF} is slightly different from \citep{siriya2025frameworkadaptivestabilisationnonlinear} in terms of dependent arguments, the addition dependence on $\hat{\theta}, \epsilon$ does not affect the inequalities of \eqref{eq:stoch_lya_ineq}, and so the same proof idea of \citep{siriya2025frameworkadaptivestabilisationnonlinear} is still applicable. Therefore, it suffices to verify that all the required assumptions in \citep{siriya2025frameworkadaptivestabilisationnonlinear} are satisfied in this paper.
    \begin{itemize}
        \item Assumption 1 in \citep{siriya2025frameworkadaptivestabilisationnonlinear}: A random variable is, by definition, a measurable function from the underlying probability space to the reals \citep{bernstein2020measure}, so the process noise $w$ and the excitory input signal $\eta$ are all measurable. Since the linear mapping $Ax+Bu$ is linear and continuous, thus it is measurable. Our system model \eqref{eq:system} is a linear mapping plus $w$, which is a summation of two measurable functions. Thus it is Borel measurable. The basis function in our setting is the identity map $(x,u)$, which is also continuous and thus Borel measurable. For MPC control law defined by solving \eqref{eq:MPCproblemfinal}, according to Berge's Maximum Theorem (i.e. Theorem 17.31 in \citep{aliprantis2006infinite}) as well as the convexity of \eqref{eq:MPCproblemfinal}, the optimal solution of $v$ is unique and a continuous function in $x$. The final control input is $u_k = Kx_k+v_k+\eta_k$, which is therefore a continuous function plus a measureable random variable. So the input is Borel measureable.
        \item Assumption 2 in \citep{siriya2025frameworkadaptivestabilisationnonlinear}: This assumption requires the process noise to be i.i.d., zero-mean and sub-Gaussian. Our Assumption~\ref{assump:disturbance} satisfies this directly.
        \item Assumption 3 in \citep{siriya2025frameworkadaptivestabilisationnonlinear}: For our problem setting, once the MPC problem is initially feasible, it stays recursively feasible as shown in Proposition \ref{proposition:feasibility}. Thus the state trajectory is bounded inside $\mathcal{Z}$ on the online phase. Together with boundedness during offline phase according to Assumption \ref{assump:boundedoffline}, this condition holds.
        \item Assumption 4 in \citep{siriya2025frameworkadaptivestabilisationnonlinear}: Our control input is bounded by construction. 
        \item Assumption 5 in \citep{siriya2025frameworkadaptivestabilisationnonlinear}: In the linear setting \eqref{eq:system}, the basis function reduces to the identity 
        mapping $(x,u)$, which satisfies the required regularity properties.
        \item Assumption 9 in \citep{siriya2025frameworkadaptivestabilisationnonlinear}: We have shown the global excitation condition in Lemma \ref{lem:RLSbound}.
        \item Assumption 10 in \citep{siriya2025frameworkadaptivestabilisationnonlinear}: The existence of the stochastic ISS function is provided by Proposition \ref{prop:SLya}.
    \end{itemize}

    Since all the assumptions are satisfied, our conclusion follows. The explicit construction of $\beta,c_4$ are available in \citep{siriya2025frameworkadaptivestabilisationnonlinear}.
    
\end{proof}

\begin{remark}
    In \eqref{eq:statebound}, the constant $c_4$ captures the effect of non-vanishing disturbances, while $\beta(k)$ includes the influence of the estimation error, which decreases monotonically according to Lemma \ref{lem:RLSbound}. Compared with the seminal work \cite{lorenzen2019robust}, where the final bound is related to the estimation error at fixed time $k_0$, our result shows that the influence of the estimation error diminishes over time. Consequently, the state ultimately converges to a set whose size is determined solely by the non-vanishing disturbances.
\end{remark}

\section{Simulation} \label{sec:simulation}
To demonstrate the effectiveness of the control method in Algorithm \ref{MPCalgorithm}, we consider the following open-loop unstable linear system
\[
x_{k+1} =
\begin{bmatrix}
0.5 & 0.2 \\
-0.3 & 1.5
\end{bmatrix} x_k
+
\begin{bmatrix}
0 \\ 0.5
\end{bmatrix} u_k
+ w_k
\]
where $w_{\max} = 0.1, w_k \overset{i.i.d}{\sim} \operatorname{Unif}([-0.1,0.1]^n), \gamma=0.01, N_0=200, \eta_{max}=0.1$. The state and input constraints are using the parameters $x_{\max}=3, u_{\max}=2$ plus the requirement $[x_k]_2 \geq -0.4$. The MPC parameters are $N=10, Q=2\times \mathbb{I}_2, R=1, K=[0.4954 -1.9476]$. The initial state is $x_0 = [2,1]^\intercal$. The predicted trajectory at the first online timestep $k_0$ as well as the tubes are shown in Fig.\ref{state_tube}. We can see that the true trajectory is within the tube around the predicted trajectory. The closed-loop state and control input is shown in Fig.\ref{input}, which satisfies all the constraints. 

\begin{figure}[htbp]
\centering
\includegraphics[scale=0.35]{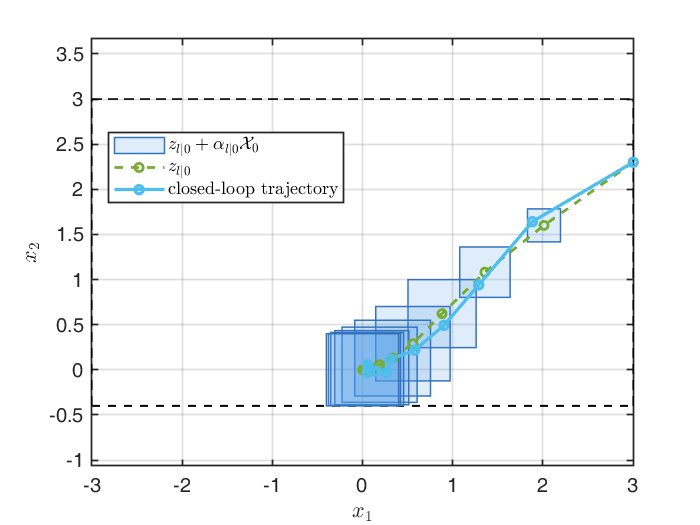}
\caption{trajectory and tube prediction at $k_0$} \label{state_tube}
\end{figure}

\begin{figure}[htbp]
\centering
\includegraphics[scale=0.35]{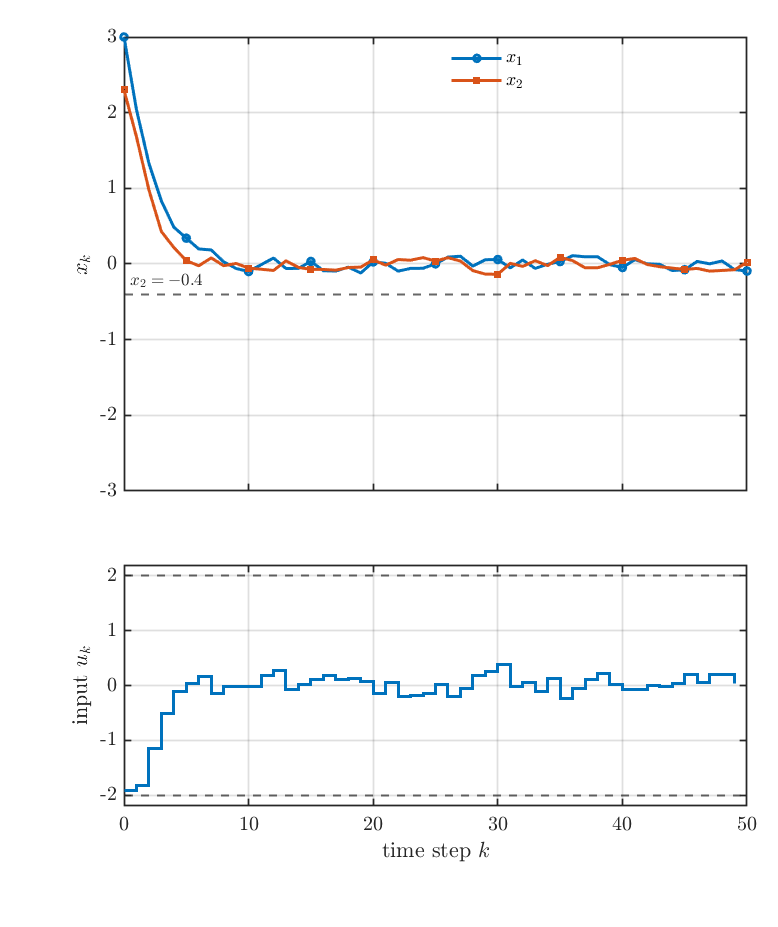}
\caption{closed-loop state and control input} \label{input}
\end{figure}


\section{Conclusion} \label{sec:conclusion}

We studied simultaneous learning and stabilization for constrained linear systems with bounded additive disturbances and parametric uncertainty. The proposed controller couples a certainty-equivalence MPC law with homothetic prediction tubes and an RLS estimator. The tube design subsumes all uncertainties to ensure robust constraint satisfaction, while the RLS yields non-asymptotic error bounds. Embedding these bounds into the MPC tightening leads to probabilistic recursive feasibility and high-probability practical stability of the closed-loop. The online problem is reformulated as a quadratic program with linear and second-order cone constraints. A numerical example is provided to illustrate these properties.

\bibliography{ifacconf}

\appendix

\section{Proof of Lemma \ref{lem:RLSbound}}
We will need the following supportive lemma, which establishes a Persistent Excitation (PE) condition of the closed-loop system. The proof of this lemma is deferred to Sec.\ref{sec:BMSB_proof}.
\begin{lemma}
Consider system~\eqref{eq:system} under Assumptions~\ref{assump:disturbance} and controller~\eqref{eq:controllaw}. There exist constants $c_{\mathrm{PE}}>0$ and $p_{\mathrm{PE}}>0$ such that, for all $x \in \mathbb{R}^n$, $\zeta \in \mathcal{S}^{n+ m-1}$, and $\theta \in \mathbb{R}^{n\times (n+m)}$,
\[
\mathbb{P}\left(
  \Big| \zeta^\top[(x+w);(\pi(x+w,\theta)+v)] \Big|^2
  \geq c_{\mathrm{PE}}
\right) \geq p_{\mathrm{PE}},
\]
where $w \overset{d}{\sim} w_k$ and $v \overset{d}{\sim} \eta_k$.
\label{lemma_BMSB}
\end{lemma}

\textbf{Proof of Lemma \ref{lem:RLSbound}:} Firstly, we can show that the system satisfies some global excitation condition (i) through Lemma \ref{lemma_BMSB}. Then applying Corollary 1 of \citep{siriya2024non} gives us the conclusion.

\section{Proof of Lemma \ref{lem:nestedparamsets}}
We conduct the proof based on the event \[
\mathcal E=\{ \|\theta^* - \bar \theta_k\|_2 \leq \bar \epsilon_k,\ \forall k\ge T_{\mathrm{burn-in}}\},
\]which holds with probability at least $1-\delta$ according to Lemma \ref{lem:RLSbound}. Recall that $\bar \epsilon$ from Proposition \ref{prop:bareps_k} is the non-asymptotic estimation bound and $T_{\text{burn-in}}$ from \ref{lem:RLSbound} is the burn-in time.

    \text{(i)} Based on the event, we have $\theta^* \in \bar\Theta_k$ for all $k \geq T_{\mathrm{burn-in}}$. Because $\bar{\Theta}_k^{\mathrm{p}}$ is a coordinate-wise tight outer approximation of $\bar\Theta_k$, we have $\theta^* \in \bar\Theta_k \in \bar{\Theta}_k^{\mathrm{p}}$ for all $k \geq T_{\mathrm{burn-in}}$. Therefore, we can conclude that $\theta^* \in \bar{\Theta}_1^{\mathrm{p}} \cap \bar{\Theta}_2^{\mathrm{p}} \cap \ldots \bar{\Theta}_k^{\mathrm{p}}$, which is equivalent to $\Theta_k$ according to \eqref{eq:Theta_k}. So $\theta^* \in \Theta_k$. Similarly, since $\Theta_k=\bar{\Theta}_1^{\mathrm{p}} \cap \bar{\Theta}_2^{\mathrm{p}} \cap \ldots \bar{\Theta}_k^{\mathrm{p}}$ for all $k$, we have $\Theta_k = \Theta_{k-1} \cap \bar{\Theta}_k^{\mathrm{p}} \subseteq \Theta_{k-1}$. Lastly, $\hat{\theta}_k \in \Theta_k$ holds because $\Theta_k$ is convex and $\hat{\theta}_k$ is a weighted sum of all the vertices of $\Theta_k$.

    (ii) Let \(d:=n(n+m)\). For each \(i\in\mathbb{N}_1^d\), define
\(
\underline{\theta}_{k,i}:=
\min_{\theta\in\Theta_k}[\operatorname{vec}(\theta)]_i,
\
\overline{\theta}_{k,i}:=
\max_{\theta\in\Theta_k}[\operatorname{vec}(\theta)]_i,
\)
and let
\(
l_k :=
\max_{i\in\mathbb{N}_1^d}
\left(
\overline{\theta}_{k,i}
-
\underline{\theta}_{k,i}
\right).
\) Since $\theta^* \in \Theta_{k}$ and the center of the set $\hat{\theta}_k \in \Theta_{k}$, so $\|\hat{\theta}_k - \theta^*\|_\infty \leq \tfrac{l_k}{2} \implies \|\hat{\theta}_k - \theta^*\|_2 \leq \sqrt{n \times (n+m)} \; \tfrac{l_k}{2} $. Since $\Theta_{k} \subseteq \bar \Theta_k^{\mathrm{p}}$, we have $\tfrac{l_k}{2} \leq \bar \epsilon_k$. So $\|\hat{\theta}_k - \theta^*\|_2 \leq \sqrt{n \times (n+m)} \;\bar \epsilon_k$.

\section{Proof of Proposition \ref{prop:predictiontube}}
We will need the following two supportive lemmas, which are useful for reformulating the original constraints into second order cone constraints.

\begin{lemma} \label{lem:supnorm}
Let $a,b \in \mathbb{R}^n$ and $\epsilon \ge 0$. Then
\begin{equation}
    \sup_{\|\Delta\|_2 \le \epsilon} a^\top \Delta b
    = \epsilon \, |a| \, |b|.
    \label{eq:supnorm}
\end{equation}
\end{lemma} 
\begin{proof}
    According to \cite[Example 3.11]{boyd2004convex}, we know that for any $\Delta$
    \[
    \sup_{|a| = 1, |b| = 1} a^\top \Delta b
        = \| \Delta \|_2.
    \]
    So for any non-zero vectors $a,b$ and any $\Delta$, we have equivalently
    \[
    \sup_{a,b \in \mathbb{R}^n} \tfrac{a^\top}{|a|} \Delta \tfrac{b}{|b|}
        = \| \Delta \|_2.
    \]
    Therefore, we have for all non-zero $a,b \in \mathbb{R}^n$ and any $\Delta \in \mathbb{R}^{n \times n}$,
    \[
    a^\top \Delta b \leq |a| \, |b| \, \|\Delta\|_2.
    \]
    Considering the constraints $\| \Delta \|_2 \leq \epsilon$, we have
    \[
a^\top \Delta b
    \leq |a| \, |b|\, \| \Delta \|_2 \leq \epsilon \, |a| \, |b|,
\]
holds for all $\Delta \in \mathbb{R}^{n \times n}, \| \Delta \|_2 \leq \epsilon$, and so
\begin{equation}
    \sup_{\|\Delta\|_2 \le \epsilon} a^\top \Delta b
    \leq \epsilon \, |a| \, |b|.
    \label{eq:lem5_1}
\end{equation}
    Now we prove that the equation can always be achieved for any non-zero $a,b \in \mathbb{R}^n$. Select $\Delta^* = \epsilon \tfrac{a}{|a|}\tfrac{b^\top}{|b|}$ such that $\|\Delta^* \|_2 \leq \epsilon$, then the equation holds
    \begin{equation}
    a^\top \Delta^* b = \epsilon \tfrac{a^\top a}{|a|}\tfrac{b^\top b}{|b|} = \epsilon |a| |b|.
    \label{eq:lem5_2}
    \end{equation}

    Lastly, we notice that \eqref{eq:supnorm} holds if either $a$ or $b$ is a zero vector. Combining further \eqref{eq:lem5_1}\eqref{eq:lem5_2} concludes the proof.

\end{proof}

\begin{lemma} \label{lem:supinterchange}
Let $\mathcal{X}$ and $\mathcal{Y}$ be non-empty sets and let $f : X \times Y \to \mathbb{R}$. Then
\[
\sup_{(x,y)\in \mathcal{X}\times \mathcal{Y}} f(x,y)
  = \sup_{x\in \mathcal{X}} \sup_{y\in \mathcal{Y}} f(x,y)
  = \sup_{y\in \mathcal{Y}} \sup_{x\in \mathcal{X}} f(x,y).
\]
\end{lemma}

\begin{proof}
We prove that
\[
\sup_{(x,y)\in \mathcal{X}\times \mathcal{Y}} f(x,y)
  = \sup_{x\in \mathcal{X}} \sup_{y\in \mathcal{Y}} f(x,y).
\]
The equality with the reversed order of supremum is analogous.

\textbf{Step 1:} 
For every $(x,y)\in \mathcal{X}\times \mathcal{Y}$ we have
\[
f(x,y) \;\le\; \sup_{y\in \mathcal{Y}} f(x,y)
             \;\le\; \sup_{x\in \mathcal{X}} \sup_{y\in \mathcal{Y}} f(x,y).
\]
Taking the supremum over all $(x,y)\in \mathcal{X}\times \mathcal{Y}$ on the left-hand side yields
\(
\sup_{(x,y)\in \mathcal{X}\times \mathcal{Y}} f(x,y)
  \;\le\; \sup_{x\in \mathcal{X}} \sup_{y\in \mathcal{Y}} f(x,y).
\)

\textbf{Step 2:} 
Fix any $x\in \mathcal{X}$. By definition of the supremum over $\mathcal{X}\times \mathcal{Y}$,
\(
\sup_{y\in \mathcal{Y}} f(x,y) \;\le\; \sup_{(x,y)\in \mathcal{X}\times \mathcal{Y}} f(x,y).
\)
Taking the supremum over all $x\in \mathcal{X}$ on the left-hand side gives
\(
\sup_{x\in \mathcal{X}} \sup_{y\in \mathcal{Y}} f(x,y)
  \;\le\; \sup_{(x,y)\in \mathcal{X}\times \mathcal{Y}} f(x,y).
\)
Finally, combining the inequalities from Step~1 and Step~2 yields
\(
\sup_{(x,y)\in \mathcal{X}\times \mathcal{Y}} f(x,y)
  = \sup_{x\in \mathcal{X}} \sup_{y\in \mathcal{Y}} f(x,y).
\)
Similarly, the second equality of the lemma holds. This completes the proof.
\end{proof}

\textbf{Proof of Proposition \ref{prop:predictiontube}}
We now proceed to prove Proposition \ref{prop:predictiontube}. Each constraints in \eqref{eq:constraints} will be reformulated using the homothetic tube parameterization \eqref{eq:homotube} and the control law parameterization \eqref{eq:controllaw}\eqref{eq:controlparameterization}. Lemmas \ref{lem:supnorm} and \ref{lem:supinterchange} are used, resulting in second-order cone constraints.

Recall that $\mathcal{H}$ is defined in \eqref{eq:Hset}, then inequality (\ref{eq:tubeconstraintsc}) is equivalent to
\begin{equation}
\begin{aligned}
& Fx + Gu_{l|k} \leq \mathbf{1}, \forall x \in \mathbb{X}_{l|k} \\
& \iff Fx + G(Kx+v_{l|k}+\eta) \leq \mathbf{1} , \forall x \in \mathbb{X}_{l|k}, \forall \eta \in \mathcal{H}\\
& \iff (F+GK)x + G(v_{l|k}+\eta)\leq \mathbf{1} , \forall x \in \mathbb{X}_{l|k}, \forall \eta \in \mathcal{H}  \\
    & \iff (F + GK)(z_{l|k}+\alpha_{l|k}x) + Gv_{l|k} + G \eta \leq \mathbf{1}, \\
    & \hspace{12em} \forall x \in \mathbb{X}_0, \forall \eta \in \mathcal{H} \\
    & \overset{\eqref{eq:temp5}\eqref{eq:temp6}}{\iff} (F + GK)z_{l|k} + G v_{l|k} + \bar{\eta} + \alpha_{l|k} \bar{f} \leq \mathbf{1}.
\end{aligned} \label{eq:tubeproof_2}
\end{equation}
The first equivalence is from the input parameterization \eqref{eq:controlparameterization} and the control law \eqref{eq:controllaw}. The second last equivalence is from the homothetic tube construction, i.e., the first line of \eqref{eq:homotube}. The last equivalence holds by taking row-wise maximization of related terms.

Inequality (\ref{eq:tubeconstraintsa}) is equivalent to (\ref{eq:reconstraints_b})  directly from the definition of the tube \eqref{eq:homotube}, and (\ref{eq:tubeconstraintsb}) is equivalent to (\ref{eq:reconstraints_c}) as shown by the following reformulation. Recall that $v$ is the number of vertices in the base polytope defined in \eqref{eq:basepolytope}, we have

\begin{equation}
    \begin{aligned}
    \mathbb{X}_{l+1|k} &\ni (\hat{A} + \hat{B} K)x + \hat{B}v_{l|k} + \hat{B} \eta + w,  \\
    & \forall x \in \mathbb{X}_{l|k}, \forall \eta \in \mathcal{H}, \forall w \in \mathcal{W}, \forall [\hat{A} \ \hat{B}] \in \Theta_k
\end{aligned} \label{eq:tubeproof_1}
\end{equation}

\[
\begin{aligned}
    \overset{\eqref{eq:homotube}}{\iff} & H_x \left((\hat{A} + \hat{B} K)x + \hat{B}v_{l|k} + \hat{B}\eta+w - z_{l+1|k}\right)  \\
    & \leq \alpha_{l+1|k} \mathbf{1},\,  \forall x \in \mathbb{X}_{l|k},\, w \in \mathcal{W},\, \forall \eta \in \mathcal{H}, \forall [\hat{A} \ \hat{B}] \in \Theta_k,
\end{aligned}
\]
\[
 \overset{\eqref{eq:temp4}}{\iff}  H_x \left((\hat{A} + \hat{B} K)(z_{l|k} + \alpha_{l|k} x^j) + \hat{B}v_{l|k}  +\hat{B}\eta - z_{l+1|k}\right)
\]
\[
 - \alpha_{l+1|k} \mathbf{1} \leq -\bar{w},\, \forall j \in \mathbb{N}_1^{v}, \, \forall \eta \in \mathcal{H}, , \forall [\hat{A} \ \hat{B}] \in \Theta_k,
\]
\[
\begin{aligned}
    \iff &\max_{\substack{
    \|\Delta_{A_k}\|_2 \leq \epsilon_k, \\
    \|\Delta_{B_k}\|_2 \leq \epsilon_k, \\
    \eta \in \mathcal{H}}} \Big[ H_x \big((\hat{A}_k + \Delta_{A_k} + (\hat{B}_k+\Delta_{B_k}) K) \times \\
    & (z_{l|k} + \alpha_{l|k} x^j) + (\hat{B}_k+\Delta_{B_k})v_{l|k} +(\hat{B}_k+\Delta_{B_k})\eta\big) \Big]_i \\
     &- H_x z_{l+1|k} - \alpha_{l+1|k} \mathbf{1} \leq -[\bar{w}]_i, \ \forall j \in \mathbb{N}_1^{v}, \forall i \in \mathbb{N}_1^{h},
\end{aligned}
\] 
\[
\begin{aligned}
     \overset{\eqref{eq:temp3}}{\iff}  & \Big[ H_x \left( (\hat{A}_k + \hat{B}_k K)(z_{l|k} + \alpha_{l|k} x^j) + \hat{B}_k v_{l|k}  \right) - H_x z_{l+1|k} \\
    & - \alpha_{l+1|k} \mathbf{1} + \bar{\eta}_{\hat{B}_k} \Big]_i
+ \max_{\substack{\|\Delta_{A_k}\|_2 \leq \epsilon_k \\ \|\Delta_{B_k}\|_2 \leq \epsilon_k\\\eta \in \mathcal{H}}} \Big[
H_x \big( \Delta_{A_k}(z_{l|k} + \alpha_{l|k} x^j) \\
&+ \Delta_{B_k} K (z_{l|k} + \alpha_{l|k} x^j) 
+ \Delta_{B_k} v_{l|k} + \Delta_{B_k} \eta \big) \Big]_i \leq -[\bar{w}]_i, \\
& \hspace{16em} \forall j \in \mathbb{N}_1^v, \forall i \in \mathbb{N}_1^{h},
\end{aligned}
\]
\[
\begin{aligned}
    \iff & H_x \left( (\hat{A}_k + \hat{B}_k K)(z_{l|k} + \alpha_{l|k} x^j) + \hat{B}_k v_{l|k} \right) - H_x z_{l+1|k} \\
    &- \alpha_{l+1|k} \mathbf{1} + \bar{\eta}_{\hat{B}_k}+  \epsilon_k \Bigl(|z_{l|k}+\alpha_{l|k}x^j| \\
&+|K(z_{l|k}+\alpha_{l|k}x^j)+v_{l|k} | + \sqrt{m} \, \eta_{max} \mathbf{1} \Bigr)
\begin{pmatrix}
|H_x^1|\\
\vdots\\
|H_x^h|
\end{pmatrix} \\
& \hspace{13em}\leq -\bar{w}, \quad \forall j \in \mathbb{N}_1^v,
\end{aligned}
\]
where the last equivalence is to be proved below, leveraging Lemma \ref{lem:supnorm} and \ref{lem:supinterchange}. 


Let
\(
\xi = z_{l|k} + \alpha_{l|k}\,x^j,
\
\zeta = K\,\xi + v_{l|k},
\)
and denote the rows of \(H_x\in\mathbb{R}^{h\times n}\) by \(H_x^i\) for \(i=1,\dots,h\). Firstly, we have
\begin{equation}
\begin{split}
\max_{\substack{\|\Delta_{B_k}\|_2\le\epsilon\\ \eta \in \mathcal{H}}}
H_x^i\,\Delta_{B_k}\,\eta
& =\max_{\eta \in \mathcal{H}} \max_{\|\Delta_{B_k}\|_2\le\epsilon }
H_x^i\,\Delta_{B_k}\,\eta \\
& = \max_{\eta \in \mathcal{H}} \epsilon\,|H_x^i|\,|\eta| \\
& =  \epsilon\,|H_x^i|\,\sqrt{m} \, \eta_{\max}
\end{split}
\end{equation}
where the first equality is from Lemma \ref{lem:supinterchange}, the second equality is from Lemma \ref{lem:supnorm} and the last equation is from the boundedness of $\mathcal{H}$. Next, applying Lemma \ref{lem:supnorm} to $\xi, \zeta$:

\[
\begin{aligned}
    &\max_{\|\Delta_{A_k}\|_2\le\epsilon}
H_x^i\,\Delta_{A_k}\,\xi
=\epsilon\,|h_i|\,|\xi|,\\
&\max_{\|\Delta_{B_k}\|_2\le\epsilon}
H_x^i\,\Delta_{B_k}\,\zeta
=\epsilon\,|h_i|\,|\zeta|.
\end{aligned}
\]
Hence for each \(i\),
\[
\begin{aligned}
    & \max_{\substack{\|\Delta_{A_k}\|_2\le\epsilon\\\|\Delta_{B_k}\|_2\le\epsilon \\ \eta \in \mathcal{H}}}
H_x^i\bigl(\Delta_{A_k}\,\xi + \Delta_{B_k}\,\zeta +  \Delta_{B_k}\,\eta \bigr) \\
& \qquad=
\epsilon\,|H_x^i|\bigl(|\xi|+|\zeta| + \sqrt{m} \, \eta_{\max}\bigr).
\end{aligned}
\]
Reassembling these row‐wise maxima into a vector gives
\[
\begin{aligned}
    &\max_{\substack{\|\Delta_{A_k}\|_2\le\epsilon\\\|\Delta_{B_k}\|_2\le\epsilon \\ \eta \in \mathcal{H}}}
H_x^i\bigl(\Delta_{A_k}\,\xi + \Delta_{B_k}\,\zeta +  \Delta_{B_k}\,\eta \bigr) \\
& \qquad =
\epsilon\,
\bigl(|\xi| + |\zeta|+ \sqrt{m} \, \eta_{\max} \bigr)
\begin{pmatrix}
|H_x^1|\\
\vdots\\
|H_x^h|
\end{pmatrix}.
\end{aligned}
\]

Equivalently, substituting back \(\xi=z_{l|k}+\alpha_{l|k}x^j\) and \(\zeta=K(z_{l|k}+\alpha_{l|k}x^j)+v_{l|k}\),
\[
\begin{aligned}
&\max_{\substack{\|\Delta_{A_k}\|_2\le\epsilon\\\|\Delta_{B_k}\|_2\le\epsilon\\ \eta \in \mathcal{H}}}
H_x\Bigl(\Delta_{A_k}(z_{l|k}+\alpha_{l|k}x^j)
+\Delta_{B_k}K(z_{l|k}+\alpha_{l|k}x^j)\\
& \hspace{12em}+\Delta_{B_k} \,v_{l|k} +\Delta_{B_k}\, \eta
\Bigr)\\
&= \max_{\eta \in \mathcal{H}} \ \epsilon \|z_{l|k}+\alpha_{l|k}x^j\|_2
\begin{pmatrix}
|H_x^1|\\
\vdots\\
|H_x^h|
\end{pmatrix} \\
& \hspace{3em}+ \epsilon|K(z_{l|k}+\alpha_{l|k}x^j)+v_{l|k}| + \sqrt{m} \, \eta_{max} \mathbf{1} 
\begin{pmatrix}
|H_x^1|\\
\vdots\\
|H_x^h|
\end{pmatrix}.
\end{aligned}
\]    
So \eqref{eq:tubeconstraintsc} is equivalent to \eqref{eq:reconstraints_c}.
\hfill $\square$

\section{Proof of Proposition \ref{proposition:feasibility}}

\begin{proof}

Claim (i) holds directly From Lemma \ref{lem:nestedparamsets}.

\textit{Claim (ii):} The idea to prove (ii) is to reuse the feasible solutions at time $k$ and then append terminal control laws to the last. If (i) holds, then the initially feasible solution is also feasible for time $k+1$. At time \( k \), for \( l \in \mathbb{N}_0^{N-1} \), let 
\( u_{l|k}(x) = Kx + v_{l|k}^* \) and 
\( \mathbb{X}_{l+1|k}^* \) be a feasible input and admissible state tube trajectory satisfying the MPC constraints (\ref{eq:tubeconstraints}) and terminal constraint 
\( \mathbb{X}_{N|k}^* \in \mathbb{X}_f \). The initially feasible sequence is 
\[
\left[ (z_{0|k}^*, \alpha_{0|k}^*,v_{0|k}^*), ..., (z_{N|k}^*, \alpha_{N|k}^*,v_{N|k}^*)\right]
\]

At time \( k+1 \) and \( l \in \mathbb{N}_0^{N-1} \), define the candidate input 
\( \tilde{u}_{l|k+1}(x) = Kx + v_{l+1|k}^* \) with 
\[
v_{N|k}^* = v_{N+1|k}^*  = 0,
\]
and candidate state tube 
\( \tilde{\mathbb{X}}_{l|k+1} = \mathbb{X}_{l+1|k}^* \). The expanded candidate solution is 
\[
\begin{aligned}
    \Big[& (z_{1|k}^*, \alpha_{1|k}^*,v_{1|k}^*),...,(z_{N|k}^*, \alpha_{N|k}^*,v_{N|k}^*),  \\
    & (\tilde{z}_{N|k+1}, \tilde{\alpha}_{N|k+1},v_{N|k+1}^*)\Big]
\end{aligned}
\]
Since \( \tilde{\mathbb{X}}_{N-1|k+1} = \mathbb{X}_{N|k}^* \in \mathbb{X}_f \) and 
\( \Theta_{k+1} \subseteq \Theta_k \subseteq \Theta_0\), by Assumption \ref{assump:terminalset}, there exists 
\( \tilde{\mathbb{X}}_{N|k+1}(\tilde{z}_{N|k+1}, \tilde{\alpha}_{N|k+1}) \subseteq \mathbb{X}_f\) satisfying 
\[
\begin{aligned}
    &Ax + B\tilde{u}_{N-1|k+1}(x) + B \eta +w \in \tilde{\mathbb{X}}_{N|k+1}, \\
&\forall \eta \in \mathcal{H}, \forall w \in \mathcal{W}, \forall x \in \tilde{\mathbb{X}}_{N-1|k+1} \subseteq \mathbb{X}_f, \forall \, \theta \in \Theta_{k+1}.
\end{aligned}
\]
By construction, \( \{\tilde{u}_{l|k+1}, \tilde{\mathbb{X}}_{l|k+1}\}_{l \in \mathbb{N}_0^{N-1}} \) satisfy the constraints (\ref{eq:tubeconstraintsb})(\ref{eq:tubeconstraintsc}), and since
\[
x_{k+1}
\in \mathbb{X}_{1|k}^* = \tilde{\mathbb{X}}_{0|k+1},
\]
constraint (\ref{eq:tubeconstraintsa}) is satisfied. 
By Proposition \ref{prop:predictiontube}, this is equivalent to feasibility of (\ref{eq:reconstraints}), and hence if $\Theta_{k+1} \subseteq \Theta_k$:
\begin{equation}
    \mathbb{D}(x_k, \Theta_k) \neq \emptyset 
\implies \mathbb{D}(x_{k+1}, \Theta_{k+1}) \neq \emptyset.
\label{eq:feasibility_temp}
\end{equation}

\textit{Claim (iii):} It is a direct consequence of \textit{Claim (i)(ii)} and Proposition \ref{prop:predictiontube} .
\end{proof} 

\section{Proof of Proposition \ref{prop:SLya}}

\begin{proof}
    
Condition the following derivation on the event 
\[
\mathcal E=\{\theta^\star\in \Theta_k,\ \Theta_{k+1}\subseteq \Theta_k,\ \forall k\ge k_0\},
\]
which satisfies $\mathbb P(\mathcal E)\ge 1-\delta$ by Lemma~\ref{lem:nestedparamsets}. All statements below hold for all $k\ge k_0$ on~$\mathcal E$.
\begin{enumerate}

    \item Due to the recursive feasibility in Proposition \ref{proposition:feasibility}, we can conclude that at any time $k\geq k_0$, the system state $x_k \in \{ x \mid \mathbb{D}(x,\Theta_k) \neq \emptyset  \}$. Also, for any $x_k \in \{ x \mid \mathbb{D}(x,\Theta_k) \neq \emptyset\}$, we have $x_{k+1} \in \{ x \mid \mathbb{D}(x,\Theta_k) \neq \emptyset\}$. This is because as long as $\mathbb{D}(x,\Theta_k) \neq \emptyset$, we can always reuse the feasible solution from time $k$ and append a terminal control law in the end, which ensures feasibility at time $k+1$. This can also be seen in \eqref{eq:feasibility_temp} by replacing $\Theta_{k+1}$ with $\Theta_k$.
    



    \item Since problem \eqref{eq:MPCproblemfinal} has a certainty-equivalence quadratic cost function, stabilizing terminal constraints and compact constraint sets, similar proof as in \citep[Lemma 4]{limon2004robust} can be invoked to show that there exists $c_2>0$ such that 
    \begin{equation}
    \label{eq:upperLya}
    \begin{aligned}
        V_N\bigl(x,\theta,\epsilon\bigr) &= J_N\bigl(x,\theta,\mathbf{v}_N^* (x,\theta,\Theta)\bigr) \\
        & < c_2 |x|^2 =: \alpha_2(|x|),
    \end{aligned}
    \end{equation}
    where $\mathbf{v}_N^* (x,\theta,\Theta)$ is the optimal solution of the MPC problem. Due to the positive definiteness of $Q$, we also have $V_N\bigl(x,\theta,\epsilon\bigr) \geq \lambda_{\min}(Q)|x|^2 =: \alpha_1(|x|).$
    
    \item We then show that the closed-loop system has an RPI set, and then prove the cost decrease condition in expectation. Note that 
    \begin{equation}
        v_{l|k+1}=v^*_{l+1|k} =: v_{k+l+1}, \forall \, l \in \mathbb{N}_0^{N-2},
        \label{eq:reusev}
    \end{equation}
    \begin{equation}
        v_{k+N}:= v_{N-1|k+1}=0
        \label{eq:terminal0v}
    \end{equation}
    is a feasible input sequence at time $k+1$. Let \( \{ \hat{x}_{l|k} \}_{l \in \mathbb{N}_1^N} \) and \( \{ \hat{x}_{l-1|k+1} \}_{l \in \mathbb{N}_1^N} \) denote the corresponding predicted state trajectories, which evolve according to \eqref{CEtrajectory} with initial conditions \( x_k \) and \( x_{k+1} \), respectively, and calculate the difference step-by-step as follow.

    Define
\[
\hat F_k = \hat A_k + \hat B_k K, \ \hat F_{k+1} = \hat A_{k+1} + \hat B_{k+1} K,
\]
and the two prediction sequences (with the same input sequence \(v_{k+i}\)):
\begin{align*}
\hat x_{0\mid k} &= x_k, \\
\hat x_{i+1\mid k} &= \hat F_k\,\hat x_{i\mid k} + \hat B_k\,(v_{k+i}+ \eta_{k+i}),\\
\hat x_{0\mid k+1} &= x_{k+1}, \\
\hat x_{i+1\mid k+1} &= \hat F_{k+1}\,\hat x_{i\mid k+1} + \hat B_{k+1}\,(v_{k+1+i}+\eta_{k+1+i}).
\end{align*}

First, for $l=1$, using the true dynamics
\[
\begin{aligned}
  x_{k+1} &= (A+BK)\,x_k + (B\,v_k+\eta_k) + w_k,\\
  \hat x_{1\mid k} &= \hat F_k\,x_k + \hat B_k\,(v_k+\eta_k),  
\end{aligned}
\]
we get
\begin{equation}
\begin{aligned}
    \delta_{1\mid k}
&= x_{k+1} - \hat x_{1\mid k}\\
& = \bigl((A+BK)-\hat F_k\bigr)\,x_k
+ (B-\hat B_k)\,(v_k+\eta_k)
+ w_k \\
& = (\Delta_{A_k}+\Delta_{B_k}K)x_k + \Delta_{B_k}(v_k+\eta_k)+w_k,
\end{aligned} \label{eq:delta1k}
\end{equation}
where \(\Delta_{A_k} = A - \hat A_k\), \(\Delta_{B_k} = B - \hat B_k\).

For recursion of \(l\ge2\), subtract the two updates:
\begin{align*}
\hat x_{l-1\mid k+1}
&= \hat F_{k+1}\,\hat x_{l-2\mid k+1} + \hat B_{k+1}\,(v_{k+l-1}+\eta_{k+l-1}),\\
\hat x_{l\mid k}
&= \hat F_k\,\hat x_{l-1\mid k} + \hat B_k\,(v_{k+l-1}+\eta_{k+l-1}).
\end{align*}
Thus
\[
\begin{aligned}
    \delta_{l\mid k} &= \hat x_{l-1\mid k+1} - \hat x_{l\mid k} \\
&= \hat F_{k+1}\,\hat x_{l-2\mid k+1}
- \hat F_k\,\hat x_{l-1\mid k} \\
& \quad+ (\hat B_{k+1}-\hat B_k)\,(v_{k+l-1}+\eta_{k+l-1}).
\end{aligned}
\]
Adding and subtracting \(\hat F_{k+1}\,\hat x_{l-1\mid k}\):
\[
\begin{aligned}
& \hat F_{k+1}\,\hat x_{l-2\mid k+1}
- \hat F_k\,\hat x_{l-1\mid k}\\
&= \hat F_{k+1}\bigl(\hat x_{l-2\mid k+1}-\hat x_{l-1\mid k}\bigr)
+ (\hat F_{k+1}-\hat F_k)\,\hat x_{l-1\mid k}\\
&= \,\hat F_{k+1}\,\delta_{l-1\mid k}
+ (\hat F_{k+1}-\hat F_k)\,\hat x_{l-1\mid k}.
\end{aligned}
\]
Hence the recursion is 
\[
\begin{aligned}
\delta_{l\mid k} &= \hat F_{k+1}\,\delta_{l-1\mid k}
+ (\hat F_{k+1}-\hat F_k)\,\hat x_{l-1\mid k} \\
& \quad +(\hat B_{k+1}-\hat B_k)\,(v_{k+l-1}+\eta_{k+l-1}).
\end{aligned}
\]
Unrolling gives
\[
\begin{aligned}
&\delta_{l\mid k} = \left(\hat F_{k+1}\right)^{l-1}\delta_{1\mid k}\\
&+ \sum_{j=1}^{l-1}
\left(\hat F_{k+1}\right)^{l-1-j}
\Big[(\hat F_{k+1}-\hat F_k) \hat x_{j\mid k}
\\
& \qquad \qquad + (\hat B_{k+1}-\hat B_k)(v_{k+j}+\eta_{k+j})\Big] \\
& \overset{\eqref{eq:delta1k}}{=}\bigl(\hat F_{k+1}\bigr)^{l-1} \Big[(\Delta_{A_k}+\Delta_{B_k}K)x_k+ \Delta_{B_k}(v_k+\eta_k)+w_k \Big]\\
& + \sum_{j=1}^{l-1}
\bigl(\hat F_{k+1}\bigr)^{l-1-j}
\Bigl[(\hat F_{k+1}-\hat F_k)\hat x_{j\mid k}
\\
& \qquad \qquad + (\hat B_{k+1}-\hat B_k)(v_{k+j}+\eta_{k+j})\Bigr] .
\end{aligned}
\]
The above together with the compactness of $x,v,w$, implies that there exists some constant $c_1>0$ such that
\begin{equation}
\begin{aligned}
    & | \delta_{l\mid k}| = | \hat x_{l-1\mid k+1} - \hat x_{l\mid k} |\\
    &\leq c_1 (\| \Delta_{A_k} \|+\| \Delta_{B_k} \| \\
    & \quad+ \| \Delta_{A_{k+1}} \|+\| \Delta_{B_{k+1}} \| + | w_k |).
    \label{eq:deltalk}
\end{aligned}
\end{equation}
Next, the proof follows a standard argument. With 
\begin{equation}
    \bar{Q}=Q+K^\intercal RK, \label{eq:Qbar}
\end{equation}
we have 
\begin{equation}
\begin{aligned}
&V_N(x_{k+1},\hat\theta_{k+1},\epsilon_{k+1})
- V_N(x_k,\hat\theta_k,\epsilon_k)\\
&\le J_N(x_{k+1},\hat\theta_{k+1},\tilde{\mathbf v}_{N\mid k})
    - V_N(x_k,\hat\theta_k,\epsilon_k)\\
&\le -\,|x_k|_Q^2 - |u_k|_R^2
    + |\hat x_{N\mid k+1}|_P^2 - |\hat x_{N\mid k}|_P^2 \\
&\quad + \sum_{l=1}^{N-1}\!\Big(|\hat x_{l\mid k+1}|_{Q}^2+|\hat u_{l\mid k+1}|_R^2
                              -|\hat x_{l\mid k}|_{Q}^2-|\hat u_{l\mid k}|_R^2\Big)\\
&= -\,|x_k|_Q^2 - |u_k|_R^2
    + |\hat x_{N\mid k+1}|_P^2 - |\hat x_{N\mid k}|_P^2 \\
&\quad + \sum_{l=1}^{N-1}\Big(|\hat x_{l\mid k+1}|_{\bar Q}^2+|v_{k+1+l}|_R^2
                              -|\hat x_{l\mid k}|_{\bar Q}^2-|v_{k+l}|_R^2\Big)\\
&= -\,|x_k|_Q^2 - |u_k|_R^2
    + |\hat x_{N\mid k+1}|_P^2 - |\hat x_{N\mid k}|_P^2 \\
&\quad + \sum_{l=1}^{N-1}\Big(|\hat x_{l\mid k+1}|_{\bar Q}^2 - |\hat x_{l\mid k}|_{\bar Q}^2\Big)
    + |v_{k+N}|_R^2 - |v_{k+1}|_R^2 \\
&\leq -\,|x_k|_Q^2 - |u_k|_R^2
    + |\hat x_{N\mid k+1}|_P^2 - |\hat x_{N\mid k}|_P^2 \\
&\quad + \sum_{l=1}^{N-1}\Big(|\hat x_{l\mid k+1}|_{\bar Q}^2 - |\hat x_{l\mid k}|_{\bar Q}^2\Big)\\ 
&= -\,|x_k|_Q^2 - |u_k|_R^2
    + |\hat x_{N\mid k+1}|_P^2 - |\hat x_{N\mid k}|_P^2 \\
&\quad + \sum_{l=2}^{N-1}\Big(|\hat x_{l-1\mid k+1}|_{\bar Q}^2 - |\hat x_{l\mid k}|_{\bar Q}^2\Big) \\
& \quad+ |\hat x_{N-1\mid k+1}|_{\bar Q}^2 - |\hat x_{1\mid k}|_{\bar Q}^2\\ 
\end{aligned} \label{eq:Vinequality1}
\end{equation}
where the first equation is because of the input parameterization \eqref{eq:controlparameterization} and the definition of \eqref{eq:Qbar}. The second equation comes from reusing the feasible solution from time $k$ as in \eqref{eq:reusev}. The third inequality is from \eqref{eq:terminal0v} and the positive-definiteness of $|v_{k+1}|_R^2$. Since $x_{N-1\mid k+1} \in \mathbb{X}_f$, we apply the terminal control law and thus have $x_{N\mid k+1}=(\hat{A}_{k+1}+\hat{B}_{k+1}K)\, x_{N-1\mid k+1}$. Based on this observation, we further do a reformulation of \eqref{eq:Vinequality1} as below:
\begin{equation}
\begin{split}
    \eqref{eq:Vinequality1}&= -\,|x_k|_Q^2 - |u_k|_R^2
    + |\hat x_{N-1\mid k+1}|_P^2 - |\hat x_{N\mid k}|_P^2 \\
&\quad + \sum_{l=2}^{N-1}\Big(|\hat x_{l-1\mid k+1}|_{\bar Q}^2 - |\hat x_{l\mid k}|_{\bar Q}^2\Big) - |\hat x_{1\mid k}|_{\bar Q}^2 \\
& \quad+ |\hat x_{N-1\mid k+1}|_{\bar Q}^2  + |\hat x_{N\mid k+1}|_P^2 - |\hat x_{N-1\mid k+1}|_P^2\\
&= -\,|x_k|_Q^2 - |u_k|_R^2
    + |\hat x_{N-1\mid k+1}|_P^2 - |\hat x_{N\mid k}|_P^2 \\
&\quad + \sum_{l=2}^{N-1}\Big(|\hat x_{l-1\mid k+1}|_{\bar Q}^2 - |\hat x_{l\mid k}|_{\bar Q}^2\Big) - |\hat x_{1\mid k}|_{\bar Q}^2 \\
& \quad+ |\hat x_{N-1\mid k+1}|_{\bar Q}^2  - |\hat x_{N-1\mid k+1}|_P^2 \\
& \quad + |\hat x_{N-1\mid k+1}|_{(\hat{A}_{k+1}+\hat{B}_{k+1}K)^\intercal P (\hat{A}_{k+1}+\hat{B}_{k+1}K)}^2\\
& \leq -\,|x_k|_Q^2 - |u_k|_R^2
    + |\hat x_{N-1\mid k+1}|_P^2 - |\hat x_{N\mid k}|_P^2 \\
&\quad + \sum_{l=2}^{N-1}\Big(|\hat x_{l-1\mid k+1}|_{\bar Q}^2 - |\hat x_{l\mid k}|_{\bar Q}^2\Big)
\end{split} \label{eq:Vinequality2}
\end{equation}
and the last inequality is due to Assumption \ref{assump:KP} and the fact that $[\hat{A}_{k+1} \ \hat{B}_{k+1}] \in \Theta_0$ under the event $\mathcal E$ such that 
\begin{align*}
    & |\hat x_{N-1\mid k+1}|_{\bar Q}^2  - |\hat x_{N-1\mid k+1}|_P^2 \\
&  + |\hat x_{N-1\mid k+1}|_{(\hat{A}_{k+1}+\hat{B}_{k+1}K)^\intercal P (\hat{A}_{k+1}+\hat{B}_{k+1}K)}^2 \leq 0.
\end{align*}

We now proceed the derivation by applying Young's inequality to \eqref{eq:Vinequality2}. We have for any $l \in \mathbb{N}_2^{N}$ and any $\rho > 0$
\begin{equation}
\begin{split}
    & |\hat x_{l-1\mid k+1}|_{\bar Q}^2 - |\hat x_{l\mid k}|_{\bar Q}^2 \overset{\eqref{eq:deltalk}}{=} |\hat x_{l\mid k}+\delta_{l \mid k}|_{\bar Q}^2 - |\hat x_{l\mid k}|_{\bar Q}^2 \\
    & \leq  (\hat x_{l\mid k}+\delta_{l\mid k})^\top \bar Q (\hat x_{l\mid k}+\delta_{l\mid k})
   - \hat x_{l\mid k}^\top \bar Q \hat x_{l\mid k} \\[2pt]
&= 2\,\hat x_{l\mid k}^\top \bar Q \delta_{l\mid k}
   + \delta_{l\mid k}^\top \bar Q \delta_{l\mid k} \\[2pt]
&\le 2\,|\hat x_{l\mid k}|_{\bar Q}\,|\delta_{l\mid k}|_{\bar Q}
   + |\delta_{l\mid k}|_{\bar Q}^2 \\[2pt]
&\le \rho\,|\hat x_{l\mid k}|_{\bar Q}^2
   + \tfrac{1}{\rho}\,|\delta_{l\mid k}|_{\bar Q}^2
   + |\delta_{l\mid k}|_{\bar Q}^2 \\[2pt]
&= \rho\,|\hat x_{l\mid k}|_{\bar Q}^2
   + \Bigl(1+\tfrac{1}{\rho}\Bigr)\,|\delta_{l\mid k}|_{\bar Q}^2,
\end{split} \label{eq:Vinequality3}
\end{equation}
where the third inequality uses Young's inequality. Similarly, we have 
\begin{equation}
    |\hat x_{N-1\mid k+1}|_{P}^2 - |\hat x_{N\mid k}|_{P}^2 \leq \rho\,|\hat x_{N\mid k}|_{P}^2
   + \Bigl(1+\tfrac{1}{\rho}\Bigr)\,|\delta_{N\mid k}|_{P}^2.
   \label{eq:Vinequality4}
\end{equation}
Integrating \eqref{eq:Vinequality3}\eqref{eq:Vinequality4} into \eqref{eq:Vinequality2} yields

\begin{equation}
\begin{aligned}
\eqref{eq:Vinequality2} &\le -|x_k|_Q^2-|u_k|_R^2 + \rho\,|\hat x_{N\mid k}|_{P}^2
   + \Bigl(1+\tfrac{1}{\rho}\Bigr)\,|\delta_{N\mid k}|_{P}^2 \\
  & \quad+\sum_{l=2}^{N-1}\rho\Bigl(|\hat x_{l\mid k}|_Q^2\Bigr) + \sum_{l=2}^{N-1}\Bigl(1+\tfrac1\rho\Bigr)|\delta_{l\mid k}|_{\bar{Q}}^2
  \\
&\le -\,|x_k|_Q^2
   + \rho\,V_N(x_k,\hat\theta_k,\epsilon_k) \\
&\quad + \sum_{l=1}^{N-1}\Big(1+\tfrac{1}{\rho}\Big)|\delta_{l\mid k}|_{\bar Q}^2
      + \Bigl(1+\tfrac{1}{\rho}\Bigr)\,|\delta_{N\mid k}|_{P}^2 \\
& \le -|x_k|_Q^2
  +\rho\,V_N\bigl(x_k,\hat\theta_k,\epsilon_k\bigr) + c_A (\| \Delta_{A_k} \|_2^2 \\
&  \quad +\| \Delta_{B_k} \|_2^2  + \| \Delta_{A_{k+1}} \|_2^2 + \| \Delta_{B_{k+1}} \|_2^2 + | w_k |^2 + |\eta_k|^2)\\
& \le -c_3 |x_k|^2 + c_A (\| \Delta_{A_k} \|_2^2+\| \Delta_{B_k} \|_2^2 \\
&  \quad + \| \Delta_{A_{k+1}} \|_2^2 + \| \Delta_{B_{k+1}} \|_2^2 + | w_k |^2+ |\eta_k|^2).
\end{aligned} \label{eq:Vinequality5}
\end{equation}
The third inequality is from combining (\ref{eq:deltalk}) and holds for some $c_A=Nc_1(1+\tfrac{1}{\rho})\max\{ \|P\|_2, \|\bar{Q}\|_2\} > 0$. The last inequality is from \eqref{eq:upperLya}. So that we can always select a small enough $\rho$ such that $\rho c_2 - \lambda_{min}(Q) \leq -c_3 < 0$ for some $c_3>0$. 

We continue to show an RPI set of the closed-loop system using the level set of the optimal value function. We continue to take the maximum of $w_k, \eta_k$ and integrate Lemma \ref{lem:RLSbound}) to yield
\begin{align*}
\eqref{eq:Vinequality5} &\le -c_3 |x_k|^2 + c_A (\| \Delta_{A_k} \|_2^2+\| \Delta_{B_k} \|_2^2 \\
&  \quad + \| \Delta_{A_{k+1}} \|_2^2 + \| \Delta_{B_{k+1}} \|_2^2 + | w_k |^2+ |\eta_k|^2) \\
& \leq -\,c_3\,|x_k|^2
  +\tfrac{1}{2}\sigma_3( \epsilon_k ^2)+\tfrac{1}{2}\sigma_3( \epsilon_{k+1} ^2)\\
  & \quad \  + c_A (|w_{\max}|^2 + |\eta_{\max}|^2) \\
  & \leq -\,c_3\,|x_k|^2
  +\sigma_3( \epsilon_k ^2)+ c_A (|w_{\max}|^2 + |\eta_{\max}|^2)\\
\end{align*}
where $\sigma_3(x)=4x$ and the last inequality is because $\epsilon_{k+1} \leq \epsilon_{k}$ according to Lemma \ref{lem:RLSbound}. Therefore, for any 
\[
\begin{aligned}
    |x_k| & \geq \tfrac{\sigma_3( \epsilon_{k_0} ^2)+ c_A (|w_{\max}|^2 + |\eta_{\max}|^2)}{c_3} \\
    & \geq \tfrac{\sigma_3( \epsilon_k ^2)+ c_A (|w_{\max}|^2 + |\eta_{\max}|^2)}{c_3}, \forall k \geq k_0
\end{aligned}
\]
we have $V_N(x_{k+1},\hat\theta_{k+1},\epsilon_{k+1})$. Denote \[\bar x_{\mathrm{RPI}}:= \tfrac{\sigma_3( \epsilon_{k_0} ^2)+ c_A (|w_{\max}|^2 + |\eta_{\max}|^2)}{c_3},\] then an RPI set can be defined as
\begin{equation}
\begin{split}
    \mathcal{X}_{\mathrm{RPI}}=\left\{ x \in \mathbb{R}^n  \, \Big| \, |x| \leq \max \left\{|x_{k_0}|,  \bar x_{\mathrm{RPI}} \right\} \right\}.
\end{split}
\end{equation}

\item We now continue to show the decrease condition of the value function in expectation. Equipping the RHS of \eqref{eq:Vinequality5} with an expectation and integrating the non-asymptotic bound from Lemma \ref{lem:RLSbound}, we arrive at
\begin{align*}
&\mathbb{E} \big[ V_N\bigl(x_{k+1},\hat\theta_{k+1},\epsilon_{k+1}\bigr) \mid x_k \big] -V_N\bigl(x_k,\hat\theta_k,\epsilon_k\bigr) \\
& \leq -\,c_3\,|x_k|^2
  +\tfrac{1}{2}\sigma_3( \epsilon_k ^2)+\tfrac{1}{2}\sigma_3( \epsilon_{k+1} ^2)+c_A \mathbb{E} (|w_k|^2 + |\eta_k|^2) \\
  & \leq -\,c_3\,|x_k|^2
  +\sigma_3( \epsilon_k ^2)+c_A \mathbb{E} (|w_k|^2 + |\eta_k|^2).
\end{align*}
Thus
\[
\begin{aligned}
    \alpha_3(|x|) &= c_3 |x|^2 \\
    \alpha_3 \circ \alpha_2^{-1} &= \tfrac{c_3}{c_2}
\end{aligned}
\]
with $\alpha_3 \circ \alpha_2^{-1}$ being linear and therefore convex.

\end{enumerate}
\end{proof}

\section{Proof of Lemma \ref{lemma_BMSB}} \label{sec:BMSB_proof}
 Before directly proving Lemma \ref{lemma_BMSB}, we present the following supportive lemmas. Lemma \ref{heavytailbounded} says that any bounded random variables with non-zero variance obeys some anti-concentration inequalities. Namely, they are bounded away from $0$ with positive probability. Lemma \ref{heavytailsub-Gaussian} establishes anti-concentration inequalities for sub-Gaussian variables. Lemma \ref{lemma_anticon_subgau} shows that the addition of sub-Gaussian RV and an independent variable owns some "heavy-tailed" properties, which relies on the results from Lemma \ref{heavytailbounded} and Lemma \ref{heavytailsub-Gaussian}. Lemma \ref{lem:variancelowerbound} derives sandwich bounds on the variance of the sum of projections related to the process noise and the control input.

\begin{lemma}
    Consider $x \in \mathbb{R}$ as a bounded random variable with zero mean with upper bound $|x| \leq x_{\text{max}}$ and lower bound of variance $\sigma^2$. Then for any constant $0<b<\sigma$, we have 
    \begin{equation}
    \mathbb{P}(|x| \geq b) \geq \tfrac{\sigma^2-b^2}{x^2_{\text{max}}}.
    \end{equation}
    \label{heavytailbounded}
\end{lemma}
\begin{proof}
    According to the definition of variance and its lower bound, we have for any constant $0<b<\sigma$
    \begin{equation}
    \begin{split}
        \sigma^2 &  \leq \mathbb{E} [x^2] \\
        & \leq \mathbb{E} \left(x^2 \big| |x| < b \right) \mathbb{P}(|x| < b) + \mathbb{E} \left( x^2 \big| |x| \geq b \right) \mathbb{P}(|x| \geq b) \\
        & \leq \mathbb{E} \left(x^2 \big| |x| < b \right) + \mathbb{E} \left( x^2 \big| |x| \geq b \right) \\
        & \leq b^2 + x^2_{\text{max}}\mathbb{P}(|x| \geq b)
    \end{split} \label{bounded}
    \end{equation}
    where the last inequality follows from Popoviciu's Inequality for bounded rv and the upper bound of $x$. Rearranging (\ref{bounded}) gives us the following result:
    \begin{equation}
        \mathbb{P}(|x| \geq b) \geq \tfrac{\sigma^2-b^2}{x^2_{\text{max}}}
    \end{equation}  
\end{proof}

\begin{lemma}
    Consider $w \in \mathbb{R}$ obeys the i.i.d. zero mean sub-Gaussian distribution with the lower bound on variance $\sigma^2$ and the variance proxy $\sigma_w^2$. For any positive constants $a,b$ satisfying $0<b<a, b<\sigma$, we have 
    \begin{equation}
    \begin{split}
        \mathbb{P} \left( |w| \geq b \right) &\geq \tfrac{\sigma^2 - 2\exp\left(\tfrac{-a^2}{2\sigma_w^2}\right) \left( a^2 + 2\sigma_w^2 \right) - b^2}{a^2}.
    \end{split}
    \end{equation}
    \label{heavytailsub-Gaussian}
\end{lemma}

\begin{proof}
    We divide the whole domain into three sections using the definition of variance as follows:
\[
    \sigma^2 \leq \mathbb{E} \left(w^2 \big| |x| < b \right) + \mathbb{E} \left( w^2 \big| |w| \geq a \right)  +  \mathbb{E} \left(w^2 \big| |w| \in [b,a] \right) \leq \sigma_w^2
\]
The proof lies in three steps: 1) Upper bound the first term on the RHS using Popoviciu's Inequality; 2) Upper bound the second term on the RHS using the tail decaying definition of sub-Gaussian; 3) Using similar rearranging from Lemma \ref{heavytailbounded} for bounded RV to yield anti-concentration results for $w$.

\begin{itemize}
    \item For the first term, choose a constant $b, 0<b<a$. According to Popoviciu's inequality
    \begin{equation}
    \begin{split}
        \mathbb{E} \left[w^2 \mid |w| \leq b\right] \leq \tfrac{(b-(-b))^2}{4} = b^2\\
    \end{split}
    \end{equation}

    \item For the second term:
    \begin{equation}
    \begin{split}
        \mathbb{E} & \left[w^2 \mid |w| \geq a, a>0 \right] \\
        & = \int_a^{+\infty} w^2f(w)dw + \int_{-\infty}^{-a} w^2f(w)dw\\
        & = \int_a^{+\infty}w^2dF(w)  + \int_{-\infty}^{-a} w^2 dF(w) \\
        & \leq 2\int_{+\infty}^a w^2 d \left(\exp(\tfrac{-w^2}{2\sigma_w^2}) \right) \\
        & = \tfrac{2}{\sigma_w^2} \int_a^{+\infty} w^3 \exp(\tfrac{-w^2}{2\sigma_w^2}) dw \\
        & = 2\exp\left(\tfrac{-a^2}{2\sigma_w^2}\right) \left( a^2 + 2\sigma_w^2 \right) \\
    \end{split}
    \end{equation}
    where $f(w)$ is the probability density function and $F(w)$ is the probability accumulative function. The first inequality is from sub-Gaussianity $\mathbb{P}[w \geq |t|] \leq \exp \left( \tfrac{-t^2}{2\sigma_w^2} \right)$.

    \item Combining the above two upper bounds together: 
    \[
    \begin{split}
        &\mathbb{E} \left[w^2 \mid |w| \in [b, a] \right] \\
        & = \mathbb{E} \left[w^2 \mid |w| \geq 0 \right] - \mathbb{E} \left[w^2 \mid |w| \geq a \right] - \mathbb{E} \left[w^2 \mid |w| \leq b \right] \\
        & \geq \sigma^2 - 2\exp \left(\tfrac{-a^2}{2\sigma_w^2} \right) \left( a^2 + 2\sigma_w^2 \right) - b^2
    \end{split}
    \]
    In order for the result to make sense, we should choose $a$ large enough and $b$ small enough to ensure $\sigma^2 - 2\exp\left(\tfrac{-a^2}{2\sigma_w^2}\right) \left( a^2 + 2\sigma_w^2 \right) - b^2 > 0$. Lastly, using similar rearranging method in Section 1 to finalize the proof. Notice that 
    \begin{equation}
    \begin{split}
        &\sigma^2 - 2\exp\left(\tfrac{-a^2}{2\sigma_w^2}\right) \left( a^2 + 2\sigma_w^2 \right) - b^2 \\
        & \leq \mathbb{E} \left[w^2 \mid |w| \in [b, a] \right] \\
        & = \int_b^{a} w^2f(w)dw + \int_{-b}^{-a} w^2f(w)dw\\
        & \leq \int_a^{b} a^2f(w)dw + \int_{-b}^{-a} a^2f(w)dw \\
        & = a^2 \mathbb{P}( |w| \in [b,a]) \\
        & \leq a^2 \mathbb{P}( |w| \geq b)
    \end{split}
    \end{equation}
    
    So for $0<b<a, b<\sigma$:
    \begin{equation}
    \begin{split}
        \mathbb{P} \left( |w| \geq b \right) \geq \tfrac{\sigma^2 - 2\exp\left(\tfrac{-a^2}{2\sigma_w^2}\right) \left( a^2 + 2\sigma_w^2 \right) - b^2}{a^2}.
    \end{split}
    \end{equation}

    \end{itemize}  
\end{proof}

\begin{lemma} (Anti-concentration Inequality for Shifted Sub-Gaussian Random Variable)
For a random variable $x \in \mathbb{R}$ obeying i.i.d zero mean sub-Guassian distribution of variance proxy $\sigma_x^2$, and the lower bound of the variance is $\sigma^2$, the random variable $x'=x+c$ where $c\in \mathbb{R}$ is a random variable that is independent from $x$, for any constants $a,b$ satisfying $0<b<a, \tfrac{7}{16}\sigma^2 - 2\exp\left(\tfrac{-a^2}{2\sigma_x^2}\right) \left( a^2 + 2\sigma_x^2 \right)>0, b<\tfrac{\sigma}{4}$, we have the following anti-concentration inequality:
\begin{equation}
\begin{split}
    &\mathbb{P}( |x'| \geq b) \\
    &\geq \min \Bigg\{ \max \left\{
            \begin{aligned}
            & 1-\exp\left(\tfrac{-\sigma^2}{32 \sigma_x^2}\right), \\
            & \tfrac{\tfrac{7}{16}\sigma^2 - 2\exp\left(\tfrac{-a^2}{2\sigma_x^2} \right) \left( a^2 + 2\sigma_x^2 \right)}{a^2}
            \end{aligned}
            \right\}, \\
& \qquad \qquad 1 - \exp\left(\tfrac{-(b-a)^2}{2\sigma_x^2}\right)\Bigg\} > 0.
\end{split}
\end{equation}
\label{lemma_anticon_subgau}
\end{lemma}
\begin{proof}
    Notice that the variance obeys $\text{var}(x')=\text{var}(x)+\text{var}(c)$ because $c$ is independent from $x$. Thus the lower bound $\text{var}(x') \geq \text{var}(x) \geq \sigma^2$ holds. We then consider several cases regarding how large is $c$. Then we witness that the proof holds similarly for $c < 0$ due to symmetry.
    
    \begin{itemize}
        \item When $0\leq c<b<a$, applying Lemma \ref{heavytailsub-Gaussian} to obtain
        \begin{equation}
            \begin{aligned}
            &\mathbb{P}( |x'| \geq b \mid 0\leq c<b<a) \\
            & \quad =  \mathbb{P}( x \geq b-c) + \mathbb{P}( x \leq -b-c)\\
            & \quad \geq \mathbb{P}( |x| \geq b+c) \\
            & \quad \geq \tfrac{\sigma^2 - 2\exp\left(\tfrac{-a^2}{2\sigma_x^2}\right) \left( a^2 + 2\sigma_x^2 \right) - (b+c)^2}{a^2} \\
            & \quad \geq \tfrac{\tfrac{3}{4}\sigma^2 - 2\exp\left(\tfrac{-a^2}{2\sigma_x^2}\right) \left( a^2 + 2\sigma_x^2 \right) }{a^2}
        \end{aligned} \label{shift3}
        \end{equation}
        where the last inequality is from the requirement that $b<\tfrac{\sigma}{4}$ and that $c<b$. Then there must exist a large enough $a>0$ such that $\tfrac{3}{4}\sigma^2 - 2\exp\left(\tfrac{-a^2}{2\sigma_x^2}\right) \left( a^2 + 2\sigma_x^2 \right)$.
        \item When $0\leq b \leq c<a$, then using the tail decay definition of sub-Gaussian, we have
        \begin{equation}
            \begin{aligned}
            &\mathbb{P}( |x'| \geq b\big| 0\leq b \leq c<a) \\
            & =  \mathbb{P}( x \geq b-c) + \mathbb{P}( x \leq -b-c) \\
            & \geq  \mathbb{P}( |x| \leq c-b)  \geq 1-\exp\left(\tfrac{-(c-b)^2}{2 \sigma_x^2}\right) 
            \label{shift1}
        \end{aligned}
        \end{equation}
        In the meantime, by applying Lemma \ref{heavytailsub-Gaussian}, the following also holds, complementing (\ref{shift1}).
        \begin{equation}
        \begin{aligned}
            &\mathbb{P}( |x'| \geq b\big| 0\leq b \leq c<a) \\
            & =  \mathbb{P}( x \geq b-c) + \mathbb{P}( x \leq -b-c) \\
            & \geq \mathbb{P}( |x| \geq b+c) \\
            & \geq \tfrac{\sigma^2 - 2\exp\left(\tfrac{-a^2}{2\sigma_x^2}\right) \left( a^2 + 2\sigma_x^2 \right) - (b+c)^2}{a^2} 
            \label{shift4}
        \end{aligned}
        \end{equation}
        Combining (\ref{shift1}) with (\ref{shift4}) to obtain
        \begin{equation}
        \begin{split}
            &\mathbb{P}( |x'| \geq b\big| 0\leq b \leq c<a) \\
            & \geq \max_{0\leq b \leq c<a} \left\{
            \begin{aligned}
            & 1-\exp\!\left(\tfrac{-(c-b)^2}{2 \sigma_x^2}\right), \\
            & \tfrac{\sigma^2 - 2\exp\left(\tfrac{-a^2}{2\sigma_x^2} \right) \left( a^2 + 2\sigma_x^2 \right) - (b+c)^2}{a^2}
            \end{aligned}
            \right\}
        \end{split} \label{shift5}
        \end{equation}
        Notice that the RHS of (\ref{shift1}) is greater than $0$ if and only if $c > b$, while for the RHS of (\ref{shift4}), there exists large enough $a>0$ such that the lower bound is greater than $0$ when $b+c< \sigma \leq \sigma_x$. Therefore, we can conclude that the RHS of (\ref{shift5}) is greater than $0$ by the following construction. Accoring to the premise, $b < \tfrac{\sigma}{4}$. When $b \leq c \leq \tfrac{\sigma}{2}$, taking the value of (\ref{shift4}) as
        \begin{equation}
        \begin{aligned}
            &\mathbb{P}( |x'| \geq b\big| 0\leq b \leq c<a) \geq \tfrac{\tfrac{7}{16}\sigma^2 - 2\exp\left(\tfrac{-a^2}{2\sigma_x^2}\right) \left( a^2 + 2\sigma_x^2 \right)}{a^2},
        \end{aligned} \label{shift6}
        \end{equation}
        which is greater than $0$ as long as $a$ is large enough. On the other hand, when $b \leq c, c>\tfrac{\sigma}{2}$, we use the value of (\ref{shift1}) as 
        \begin{equation}
        \begin{aligned}
            &\mathbb{P}( |x'| \geq b\big| 0\leq b \leq c<a)  \geq 1-\exp\left(\tfrac{-\sigma^2}{32 \sigma_x^2}\right) > 0.
        \end{aligned} \label{shift7}
        \end{equation}

        \item If $c\geq a > b$, first notice that as long as $c>b$, according to the lighter-tailedness of  sub-Gaussian distribution, we have:
        \begin{equation}
        \begin{split}
            &\mathbb{P}( |x'| \geq b \big| c\geq a) \geq \mathbb{P}( |x| \geq c-b) \\
            & \quad \geq 1 - \exp(\tfrac{-(b-c)^2}{2\sigma_x^2}) ,\quad 0<b<a, b<\sigma
        \end{split}
        \end{equation}
        For $c\geq a$, we have the following inequality:
        \begin{equation}
        \begin{split}
            &\mathbb{P}( |x'| \geq b\big| c\geq a) \geq 1 - \exp(\tfrac{-(b-c)^2}{2\sigma_x^2}) \\
            & \geq 1 - \exp\left(\tfrac{-(b-a)^2}{2\sigma_x^2}\right), 0<b<a, b<\sigma, c \geq a
        \end{split} \label{shift2}
        \end{equation}
    \end{itemize}

Finally, combining (\ref{shift3})(\ref{shift6})(\ref{shift7})(\ref{shift2}) achieves the conclusion. Similar results can be proved for $c < 0$. 

\end{proof}

\begin{lemma}(Sandwich bounds on the variance of a sum of projections)
Let $w\in\mathbb{R}^{n}$ and $\pi\in\mathbb{R}^{n}$ be random vectors with finite second moments, and let $\zeta_1\in\mathbb{R}^{n}$, $\zeta_2\in\mathbb{R}^{n}$ be deterministic vectors. Denote the covariance matrices
\[
\Sigma_w := \operatorname{Cov}(w,w),\
\Sigma_\pi := \operatorname{Cov}(\pi,\pi),\
\Sigma_{w,\pi} := \operatorname{Cov}(w,\pi).
\]
Then the scalar random variable $Y := \zeta_1^\top w + \zeta_2^\top \pi$ satisfies
\[
\operatorname{Var}(\zeta_1^\top w + \zeta_2^\top \pi)\ \ge\
\bigl(\sqrt{\zeta_1^\top \Sigma_w\,\zeta_1} - \sqrt{\zeta_2^\top \Sigma_\pi\,\zeta_2}\bigr)^2,
\]
\[
\operatorname{Var}(\zeta_1^\top w + \zeta_2^\top \pi)\ \le\
\bigl(\sqrt{\zeta_1^\top \Sigma_w\,\zeta_1} + \sqrt{\zeta_2^\top \Sigma_\pi\,\zeta_2}\bigr)^2.
\]
\label{lem:variancelowerbound}
\end{lemma}

\begin{proof}
Denote $\operatorname{Var}(\zeta_1^\top w) = \zeta_1^\top \Sigma_w\zeta_1,\,
\operatorname{Var}(\zeta_2^\top \pi) = \zeta_2^\top \Sigma_\pi\zeta_2,\, $ $
\operatorname{Cov}(\zeta_1^\top w,\zeta_2^\top \pi)= \zeta_1^\top \Sigma_{w,\pi}\,\zeta_2.$ By the variance–covariance decomposition,
\[
\operatorname{Var}(\zeta_1^\top w+\zeta_2^\top \pi)=\operatorname{Var}(\zeta_1^\top w) + \operatorname{Var}(\zeta_2^\top \pi) + 2\operatorname{Cov}(\zeta_1^\top w,\zeta_2^\top \pi).
\]
Apply Cauchy–Schwarz to the covariance of the scalars $\zeta_1^\top w,\zeta_2^\top \pi$:
\[
\begin{aligned}
\bigl|\operatorname{Cov}(\zeta_1^\top w,\zeta_2^\top \pi)\bigr| &\le \sqrt{\operatorname{Var}(\zeta_1^\top w)\operatorname{Var}(\zeta_2^\top \pi)} \\
& = \sqrt{(\zeta_1^\top \Sigma_w\,\zeta_1)(\zeta_2^\top \Sigma_\pi\,\zeta_2)}.
\end{aligned}
\]
Therefore,
\[
\begin{aligned}
& \operatorname{Var}(\zeta_1^\top w+\zeta_2^\top \pi) \\
    &\geq \operatorname{Var}(\zeta_1^\top w)+\operatorname{Var}(\zeta_2^\top \pi)-2\sqrt{\operatorname{Var}(\zeta_1^\top w)\operatorname{Var}(\zeta_2^\top \pi)},
\end{aligned}
\]
and that 
\[
\begin{aligned}
& \operatorname{Var}(\zeta_1^\top w+\zeta_2^\top \pi) \\
    &\leq \operatorname{Var}(\zeta_1^\top w)+\operatorname{Var}(\zeta_2^\top \pi)+2\sqrt{\operatorname{Var}(\zeta_1^\top w)\operatorname{Var}(\zeta_2^\top \pi)}.
\end{aligned}
\]
\end{proof}

Lastly, considering the closed-loop system as a shifted version of $w$ and $v$, resorting to Lemma \ref{lemma_anticon_subgau}, we prove that the system (\ref{eq:system}) with the controller (\ref{eq:controllaw}) satisfies some anti-concentration conditions.

\textbf{Proof of Lemma \ref{lemma_BMSB}:} Denote $\zeta = [\zeta_1^\intercal  \ \zeta_2^\intercal ]^\intercal $ with $\zeta_1 \in \mathbb{R}^n, \zeta_2 \in \mathbb{R}^n$. The intuition of the proof is that when $| \zeta_2 |$ is large, we focus on $v$ to provide the results. In this case, the variable is considered a sub-Gaussian process ($\zeta_2^\intercal  v$) shifted by $\zeta_1^\intercal w + \zeta_1^\intercal  x + \zeta_2^\intercal \pi(x+w,\theta)$ which is independent from $v$, enabling us to apply Lemma \ref{lemma_anticon_subgau}. Similarly, for the second situation where $| \zeta_2 |$ is small and $|\zeta_1|$ is large, we utilize $w$. In this case, $\zeta_2$ should be small enough so that it can be overpowered by $\zeta_1 w$. We now proceed to the detailed steps.

\begin{enumerate}
\item When $1 \geq | \zeta_2 | \geq \tfrac{eb}{\sqrt{u_{\mathrm{max}}^2+e^2b^2}}, e \in (0,1)$ and $e$ is user-specified: \\

Consider the original process as a sub-Gaussian process ($\zeta_2 v$) shifted by $\zeta_1^\intercal w + \zeta_1^\intercal  x + \zeta_2^\intercal \pi(x+w,\theta)$ which is independent from $v$. Notice that $v$ is an i.i.d bounded uniform distribution and that the variance of each entry is $\sigma_V^2 = C^2/3, C=\eta_{max}$. Based on Lemma \ref{lemma_anticon_subgau}, we start working on the original variable. Because $1 \geq | \zeta_2 | \geq \tfrac{eb}{\sqrt{u_{\mathrm{max}}^2+e^2b^2}}$, so
\begin{equation}
\begin{aligned}
&\mathbb{P}\Big( 
    \big| \zeta_1^\intercal  w 
    + \zeta_1^\intercal  x
    + \zeta_2^\intercal  \pi(x+w,\theta) 
    + \zeta_2^\intercal  v \big|  \geq \tfrac{e b^2}{\sqrt{u_{\mathrm{max}}^2+e^2 b^2}} 
\Big) \\
&\geq\;
\mathbb{P}\Big( 
    \big| \zeta_1^\intercal  w 
    + \zeta_1^\intercal  x
    + \zeta_2^\intercal  \pi(x+w,\theta) 
    + \zeta_2^\intercal  v \big|  \geq | \zeta_2 | b
\Big).
\end{aligned}
\label{2}
\end{equation}
Denote $V':=\zeta_1^\intercal  w + \zeta_1^\intercal   x +\zeta_2^\intercal  \pi(x+w,\theta)$ and so
\begin{equation}
    \begin{split}
        &\mathbb{P}\Big( \left| \zeta_1^\intercal  w + \zeta_1^\intercal   x+\zeta_2^\intercal  \pi(x_{t+1},\theta) + \zeta_2^\intercal  v  \Big| \geq | \zeta_2 | b \right) \\
        &= \mathbb{P}( |\zeta_2^\intercal  v+  V'| \geq | \zeta_2 | b) \\
        &= \mathbb{P}( | v+ \zeta_2^\intercal  V'/|\zeta_2| \,| \geq b) \\
    \end{split} \label{1}
\end{equation}
where the second equality is from the fact that sum of independent sub-Gaussian rvs is still a sub-Gaussian rv. The lower bound of (\ref{1}) can be obtained by applying Lemma \ref{lemma_anticon_subgau}. Then, combining (\ref{2}) and (\ref{1}) yields

\begin{equation}
\begin{aligned}
&\mathbb{P}\Big(
    \big| 
        \zeta_1^\intercal  w
        + \zeta_1^\intercal  x
        + \zeta_2^\intercal  \pi(x+w,\theta)
        + \zeta_2^\intercal  v \big| \geq \tfrac{e b^2}{\sqrt{u_{\mathrm{max}}^2 + e^2 b^2}}
\Big) \\
&\geq \min \Bigg\{ \max \left\{
            \begin{aligned}
            & 1-\exp\left(\tfrac{-\sigma_V^2}{32 \sigma_V^2}\right), \\
            & \tfrac{\tfrac{7}{16}\sigma_V^2 - 2\exp\left(\tfrac{-a^2}{2\sigma_V^2}\right) \left( a^2 + 2\sigma_V^2 \right)}{a^2}
            \end{aligned}
            \right\}, \\
& \qquad \qquad 1 - \exp\left(\tfrac{-(b-a)^2}{2\sigma_V^2}\right)\Bigg\} > 0,  0 < b < a,\, b<\tfrac{\sigma_V}{4},\\
&\tfrac{7}{16}\sigma_V^2 - 2 \exp\left(-\tfrac{a^2}{2\sigma_V^2}\right)(a^2 + 2\sigma_V^2) > 0
\end{aligned}
\label{BMSB1}
\end{equation}

\item When $0 \leq | \zeta_2 | < \tfrac{eb}{\sqrt{u_{\mathrm{max}}^2+e^2b^2}},$ $\tfrac{u_{\mathrm{max}}}{\sqrt{u_{\mathrm{max}}^2+e^2b^2}} <  |\zeta_1| \leq 1 , e \in (0,1)$: \\

Consider the process $\zeta_1^\intercal  w + \zeta_1^\intercal   x+\zeta_2^\intercal  \pi(x+w,\theta) + \zeta_2^\intercal  v $ as a shifted version of $\zeta_1^\intercal  w + \zeta_2^\intercal  \pi(x+w,\theta)$ by $\zeta_1^\intercal x+\zeta_2^\intercal  v$. Denote $W' :=\zeta_1^\intercal   x+ \zeta_2^\intercal  v$, which is independent from $w$. 

We first show that $\zeta_1^\intercal  w + \zeta_2^\intercal  \pi(x+w,\theta)$ is equivalent to a shifted sub-Gaussian variable with a positive lower bound on the variance. Note that $\pi(x+w,\theta)$ is uniformly bounded random variable, and therefore obeys sub-Gaussian distribution. Denote its mean value as $\bar{\pi} \leq u_{\text{max}}-C$. The variance of $\zeta_1^\intercal  w + \zeta_2^\intercal  \pi(x+w,\theta)$ can be lower bounded by 
\begin{equation}
\begin{split}
    & \operatorname{Var}(\zeta_1^\top w + \zeta_2^\top \pi(x+w,\theta)) \\
    &\ge
\bigl(\sqrt{\zeta_1^\top \Sigma_w\,\zeta_1} - \sqrt{\zeta_2^\top \Sigma_\pi\,\zeta_2}\bigr)^2 \\
& \ge (\underline{\sigma}_W |\zeta_1| - (u_{\text{max}}-C) |\zeta_2|)^2 \\
&\ge \tfrac{u_{\mathrm{max}}^2}{u_{\mathrm{max}}^2+e^2b^2} (\underline{\sigma}_W-eb)^2
\end{split}
\end{equation}
where the first inequality is from Lemma \ref{lem:variancelowerbound}, the second inequality is from Assumption \ref{assump:disturbance} and the Popoviciu's inequality on bounded variables. Therefore,
\begin{equation}
\begin{split}
    & \operatorname{Var}(\zeta_1^\top w / |\zeta_1 | + \zeta_2^\top \pi(x+w,\theta)/ |\zeta_1 | )  \\
    &\ge
\bigl(\sqrt{\zeta_1^\top \Sigma_w\,\zeta_1} - \sqrt{\zeta_2^\top \Sigma_\pi\,\zeta_2}\bigr)^2/ |\zeta_1 |^2 \\
& \ge (\underline{\sigma}_W - (u_{\text{max}}-C) |\zeta_2|/ |\zeta_1 |)^2 \\
&\ge (\underline{\sigma}_W-\tfrac{eb}{u_{\mathrm{max}}})^2
\end{split}
\end{equation}
If $b<\tfrac{\underline{\sigma}_W}{4}$ and $e\leq 2u_{\mathrm{max}}$, then further we have
\[
\begin{aligned}
    &\operatorname{Var}(\zeta_1^\top w / |\zeta_1 | + \zeta_2^\top \pi(x+w,\theta)/ |\zeta_1 | )\geq \tfrac{\underline{\sigma}_W^2}{4}
\end{aligned}
\]
Similarly, if $b<\tfrac{\underline{\sigma}_W}{4}$, the upper bound on the variance can be derived as 
\begin{equation}
\begin{split}
    & \operatorname{Var}(\zeta_1^\top w / |\zeta_1 | + \zeta_2^\top \pi(x+w,\theta)/ |\zeta_1 | ) \\
    &\le
\bigl(\sqrt{\zeta_1^\top \Sigma_w\,\zeta_1} + \sqrt{\zeta_2^\top \Sigma_\pi\,\zeta_2}\bigr)^2 / |\zeta_1 |^2 \\
& \le (\underline{\sigma}_W + (u_{\text{max}}-C) |\zeta_2|/ |\zeta_1 |)^2 \\
& \le 4 \max\{\underline{\sigma}_W^2, e^2b^2 \}  \le 4 \underline{\sigma}_W^2
\end{split}
\end{equation}

Denote $(\underline{\sigma}_w')^2:= \tfrac{\underline{\sigma}_W^2}{4}, (\sigma_w')^2:= 4 \underline{\sigma}_W^2$.Then the following condition can be proved similar to the first situation and by Lemma \ref{lemma_anticon_subgau}.

\begin{equation}
\begin{aligned}
&\mathbb{P} \!\Big(
    |\zeta_1^\intercal  w + \zeta_1^\intercal   x+\zeta_2^\intercal  \pi(x+w,\theta) + \zeta_2^\intercal  v|  \geq \tfrac{u_{\mathrm{max}}b}{\sqrt{u_{\mathrm{max}}^2 + e^2 b^2}}
\Big) \\
&\geq \mathbb{P}\big( |\zeta_1 w +\zeta_2^\intercal  \pi(x+w,\theta) + W'| \geq |\zeta_1| b \big) \\
& =  \mathbb{P}( | w +\zeta_2^\intercal  \pi(x+w,\theta)/|\zeta_1| +  W'/|\zeta_1| \, | \geq b ) \\
&\geq 
\min \Bigg\{ \max \left\{
            \begin{aligned}
            & 1-\exp\left[\tfrac{-(\underline{\sigma}_w')^2}{32 (\sigma_w')^2}\right], \\
            & \tfrac{\tfrac{7}{16}(\underline{\sigma}_w')^2 - 2\exp\left(\tfrac{-a^2}{2(\sigma_w')^2}\right) \left[ a^2 + 2(\sigma_w')^2 \right]}{a^2}
            \end{aligned}
            \right\}, \\
& \qquad \qquad 1 - \exp\left[\tfrac{-(b-eb-a)^2}{2(\sigma_w')^2}\right] \Bigg\},  0 < b < a,\, b<\tfrac{\underline{\sigma}_W}{4}, \\
&\tfrac{7}{16}(\underline{\sigma}_w')^2 - 2\exp\left(\tfrac{-a^2}{2(\sigma_w')^2}\right) \left( a^2 + 2(\sigma_w')^2 \right) > 0
\end{aligned}
\label{BMSB2}
\end{equation}

\end{enumerate} 
Taking the union bound of (\ref{BMSB1}) and (\ref{BMSB2}) yields the following.

\begin{equation}
\begin{aligned}
&\mathbb{P}\left( \left| \zeta^\intercal  Z_{t+1} \,\middle|\, \mathcal{F}_t \right| 
\geq c_{\mathrm{PE}} \right) \geq p_{\mathrm{PE}}, \  \forall \zeta \in \mathcal{S}^1 \\
&c_{\mathrm{PE}} = \min \left\{
\tfrac{(u_{\mathrm{max}}{-}C)b}{\sqrt{(u_{\mathrm{max}}{-}C)^2 + e^2b^2}},\;
\tfrac{eb^2}{\sqrt{(u_{\mathrm{max}}{-}C)^2 + e^2b^2}} \right\} \\[0.4em]
&p_{\mathrm{PE}} = \min \left\{ \eqref{BMSB1},\eqref{BMSB2} \right\}.
\end{aligned}
\end{equation}

\end{document}